\documentclass{article}
\usepackage{graphicx} 
\usepackage[dvipsnames]{xcolor}
\usepackage[colorlinks=true,linkcolor=blue,citecolor=blue,urlcolor=blue,bookmarks=true]{hyperref}
\usepackage{schmidhuber}
\usepackage{bm}

\newtheorem{assumption}[theorem]{Assumption}

\usepackage{tocloft}
\usepackage{multicol}

\newcommand{\pr}[1]{\mathrm{Pr}\left[#1\right]}
\newcommand{\lr}[1]{\left(#1\right)}
\newcommand{\EE}[1]{\mathbf{E}\left[#1\right]}
\newcommand{\cC}{\mathcal{C}}

\newcommand{\cK}{\mathcal{K}}
\newcommand{\cKt}{\Tilde{\mathcal{K}}}

\newcommand{\cP}{\mathcal{P}}
\newcommand{\cS}{\mathcal{S}}
\newcommand{\cT}{\mathcal{T}}
\newcommand{\cU}{\mathcal{U}}

\newcommand{\PP}{\mathbb{P}}
\newcommand{\QQ}{\mathbb{Q}}

\newcommand\Sym{\mathrm{Sym}}
\newcommand{\bern}{\mathrm{Bern}}
\newcommand{\disj}{\mathrm{disj}}
\newcommand{\ag}{\mathrm{ag}}
\newcommand{\dis}{\mathrm{dis}}
\newcommand{\unif}{\mathrm{Unif}}

\newcommand{\diag}{\mathrm{diag}}

\declaretheorem[sibling=theorem]{problem}
\crefname{problem}{Problem}{Problems}
\Crefname{problem}{Problem}{Problems}

\newif\ifanon\anonfalse

\title{Quartic quantum speedups for community detection}
\ifanon
    \author{}
    \date{}
\else
    \author{
    Alexander Schmidhuber\footnote{alexsc@mit.edu, contributed equally}\\MIT \and
    Alexander Zlokapa\footnote{azlokapa@mit.edu, contributed equally}\\MIT
    }
\fi 

\date{October 9, 2025}
\begin{document}
\maketitle
\begin{abstract}
Community detection is a foundational problem in data science. Its natural extension to \emph{hypergraphs} captures higher-order correlations beyond pairwise interactions.
In this work, we develop a quantum algorithm for hypergraph community detection that achieves a quartic quantum speedup over the best known classical algorithm, along with superpolynomial savings in space.
Our algorithm is based on the Kikuchi method, which we extend beyond previously considered problems such as Tensor PCA and $p$XORSAT to a broad family of generalized stochastic block models. To demonstrate (near) optimality of this method, we prove matching lower bounds (up to logarithmic factors) in the low-degree framework, showing that the algorithm saturates a smooth statistical-computational tradeoff. The quantum speedup arises from a quantized version of the Kikuchi method and is based on the efficient preparation of a guiding state correlated with the underlying community structure. Our work suggests that prior quantum speedups using the Kikuchi method are sufficiently robust to encompass a broader set of problems than previously believed; we conjecture that a quantity known as marginal order characterizes the existence of these quantum speedups.

\end{abstract}
\clearpage
\tableofcontents
\clearpage

\section{Introduction}

Developing new super-quadratic quantum speedups is the central challenge of quantum algorithms research \cite{babbush2021focus,hoefler2023disentangling}. Particularly compelling targets are domains of high practical value, such as machine learning and combinatorial optimization, where genuine quantum advantages could yield broad practical impact.

One of the oldest and most widely applicable tasks in machine learning is community detection: first formalized in 1983 with stochastic block models (SBMs)~\cite{holland1983stochastic}, it provides a generic setting to analyze clustering algorithms for data on graphs.
In its original form, an SBM defines a family of random graphs where vertices are labeled by their communities. If a pair of vertices belongs to the same community, an edge is assigned with probability $\rho$; otherwise, an edge is assigned with probability $\rho'$. The simplest task is then to detect if $\rho=\rho'$ or $\rho \neq \rho'$. The signal-to-noise ratio (SNR) $\lambda$ is a quantity that goes to zero when $\rho = \rho'$ and increases as they separate.

Many relational datasets have interactions beyond pairwise relationships, necessitating a generalization of the SBM definition. A \emph{hypergraph} SBM (HSBM) assigns a hyperedge to a set of vertices based on the community labels within the set. Numerous practical settings require community detection on hypergraphs, including protein-protein interaction networks~\cite{xiong2005identification}, gene regulatory networks~\cite{michoel2012alignment}, and other applications in biology, social networks, and computer vision~\cite{momeni2017lotka,sanchez2019high,wood2012mechanism,wood2014uncovering,tekin2018prevalence,grilli2017higher,bairey2016high,ji2022fc,ma2021accurately,ghoshal2009random,zlatic2009hypergraph,govindu2005tensor}.

An HSBM assigns a hyperedge with larger probability to a set of vertices mostly belong to the same community, and with smaller probability if they belong to many different communities. Typically, the task of community detection might have several phases of computational complexity. (Here, we imagine having $n$ vertices and a constant-order hypergraph.)
\begin{itemize}
    \item \emph{Easy phase.} Above some SNR threshold ($\beta > \beta_1$), a polynomial-time algorithm can detect communities.
    \item \emph{Hard phase.} For an intermediate range of SNRs ($\beta_0 < \beta < \beta_1$), a superpolynomial-time algorithm detects communities. It may have an \emph{SNR-computation tradeoff}, where the cost of the algorithm smoothly interpolates between polynomial and exponential as the SNR decreases.
    \item \emph{Impossible phase.} Below some SNR ($\beta < \beta_0$), it is information-theoretically impossible to detect communities.
\end{itemize}
This phase diagram is common across several planted inference problems on tensors, such as Tensor PCA and planted $p$XORSAT~\cite{WAM19,hastings2006community}; one may expect that an HSBM has the same phase diagram. Somewhat surprisingly, whether or not a hard phase exists depends on the specific choice of bias between hyperedges. For the simplest choice where a hyperedge is assigned with larger probability if all the vertices in the candidate hyperedge have the same label, polynomial classical algorithms exist to detect communities in HSBMs right up to the information theoretic threshold~\cite{angelini2015spectral,kim2018stochastic,zhang2022exact,stephan2024sparse}. No hard phase (or SNR-computation tradeoff) occurs, regardless of the hypergraph's order.

The absence of a hard phase for the simple choice of HSBM occurs for a fundamental reason: the marginals of only \emph{two} vertices suffice to solve community detection in this simple HSBM~\cite{angelini2015spectral}. (In contrast, the problems of $p$-order Tensor PCA and planted $p$XORSAT require $p$-body marginals to recover the planted signal.) Hence, we refer to the aforementioned simple HSBM as a 2-marginal HSBM.

In many of the real-world applications listed above, two-body marginals do not capture all the information in the dataset: XOR-like interactions and other phenomena undetectable by pairwise marginals are well-known to occur in gene regulation and protein-protein networks ~\cite{buchler2003schemes,tagkopoulos2008predictive,moore2009epistasis,ramadan2004hypergraph,agrawal2018large,kim2006relating,kiel2013integration,franzosa2011structural}. In this work, we focus on HSBMs where community detection requires more than two-body marginals. We give a basic toy model of such an HSBM and show a quantum speedup, providing concrete evidence for the following more general conjecture.
\begin{quote}
    \centering
    \emph{Hypergraph SBMs whose clusters are undetectable from pairwise marginals \\exhibit a super-quadratic polynomial quantum speedup for community detection.}
\end{quote}
Informally, the HSBM we study assigns hyperedges with high probability to a set of vertices sharing the same label, and with low probability otherwise, chosen such that marginals smaller than $p$-wise do not see the community structure. For $k$ communities, the probability of assigning a hyperedge in the $p$-marginal HSBM is equal to that of the 2-marginal HSBM up to an $O(1/k)$ correction.

Our classical and quantum algorithms are based on the Kikuchi method \cite{WAM19}, which was first proposed for the tasks of Tensor PCA and $p$XORSAT~\cite{WAM19} due to their technical similarity: $p$XORSAT can be recast as a variant of the spiked tensor problem with sparse Rademacher observations instead of Gaussian ones. Since then, all further applications of the Kikuchi algorithm have been restricted to settings directly related to CSP refutation~\cite{guruswami2022algorithms,hsieh2023simple,alrabiah2023near,hsieh2024small}. It has thus been unclear if the Kikuchi algorithm is useful for a wider scope of problems, as well as if the super-quadratic quantum speedup would survive in those potential settings. Our work represents a new application for the classical (and quantum) Kikuchi algorithm in community detection. This is complemented by lower bounds for the $p$-marginal HSBM that suggest a statistical-computational gap based on the recent work of~\cite{kunisky2024low}. More explicitly, we show the following.
\begin{itemize}
    \item \textbf{Lower bounds against low coordinate degree functions.} For any $p$-marginal generalized SBM (including the $p$-marginal HSBM), we show lower bounds that establish a hard phase with an SNR-computation tradeoff for low coordinate degree functions, which generalize low-degree polynomial algorithms.
    \item \textbf{Tight Kikuchi algorithm.} We develop a classical Kikuchi algorithm that matches the lower bounds throughout the entire hard phase up to a logarithmic factor; our algorithm applies to any even hypergraph order $p\geq 2$.
    \item \textbf{Quartic quantum speedup.} We show that a guiding state can be efficiently prepared and that it has sufficient overlap to obtain a quartic speedup over the classical Kikuchi method for the task of strong detection between a non-trivial and null $p$-marginal HSBM. The algorithm produces the quartic speedup with high probability over random $p$-marginal HSBM instances for even $p \geq 4$ and any number of communities $k \geq 2$.
\end{itemize}

The form of the quantum algorithm is similar to previous works on Tensor PCA~\cite{hastings2020classical} and $p$XORSAT \cite{schmidhuber2025quartic}. A Kikuchi Hamiltonian and a guiding state are constructed from the hypergraph adjacency matrix; quantum phase estimation certifies a bound on the spectral norm of the Kikuchi matrix and prepares a ground eigenstate. We prove that the resulting bound on the spectral norm detects the presence of communities. 

Lastly, we comment on a recent result \cite{gupta2025classical} describing a non-spectral classical algorithm for Planted $p$XORSAT, improving quadratically over \cite{WAM19} provided that the locality parameter $p$ is a large constant. This reduces the quartic speedup in \cite{schmidhuber2025quartic} to quadratic in the parameter regime of large constant $p$. (A superpolynomial quantum
space advantage remains for all $p$.) Because the regime of practical interest for $p$XORSAT, Tensor PCA, and hypergraph community detection is small constant $p$ (say, $p=4$)~\cite{tagkopoulos2008predictive,moore2009epistasis,ramadan2004hypergraph}, an analogous result for our work, if possible, would not dequantize our speedup in the natural parameter setting. Nonetheless, it remains open if further classical techniques exist to improve the $p$ dependence of
\cite{gupta2025classical}, or if the non-spectral approach of \cite{gupta2025classical} can be quantized to reestablish a quartic
quantum speedup for all $p$. 

\begin{remark}
\label{remark:remove_log}
For Tensor PCA, a recent refined analysis of the spectral Kikuchi algorithm  has reduced its computational cost (for any given SNR) by a superpolynomial factor~\cite{kothari2025smoothtradeofftensorpca}. This improvement likely extends to the spectral Kikuchi algorithm for $p$XORSAT and community detection, which would in particular imply that the spectral Kikuchi algorithm \cite{WAM19} outperforms the non-spectral algorithm in \cite{gupta2025classical} by a superpolynomial margin -- at least under the current analysis of \cite{gupta2025classical}. Since the quantum Kikuchi algorithm in \cite{schmidhuber2025quartic} is unaffected by the improved analysis, this would re-establish a quartic quantum speedup throughout. It is therefore an interesting direction for further research to see whether the analysis of \cite{gupta2025classical} can be similarly improved.
\end{remark}

\section{Technical summary}
\label{sec:technical_summary}
Our main results apply to the following model for communities on a hypergraph.

\begin{definition}[$p$-marginal HSBM with $k$ communities]\label{def:pbernsbm}
Let $p, k, n \geq 2$ and let $a\in [k]^p$. Define
\begin{align}
    f(a_1, \dots, a_p) = \sum_{i=1}^k \prod_{j=1}^p \lr{1_{a_j=i} - \frac{1}{k}}, \qquad 0 < \theta_0 < 1/2, \qquad 0 < \epsilon < \theta_0,
\end{align}
where $\theta_0, \epsilon$ may scale nontrivially with $n$.
For each $a \in [k]^p$, define
\begin{align}
    \mu_a = \bern(\theta_0 + \epsilon f(a)), \qquad \mu_\mathrm{avg} = \bern(\theta_0).
\end{align}
The $p$-marginal HSBM$(n,k,\theta_0,\eps)$ then consists of the following two probability measures over hypergraphs $\bY$ with vertex set $V = [n]$ and hyperedge sets $E$ defined as follows.
\begin{enumerate}
    \item Under $\QQ$, for each $S \in \binom{[n]}{p}$ draw hyperedges $Y_S \sim \mu_\mathrm{avg}$ independently and add $S$ to $E$ if $Y_S = 1$.
    \item Under $\PP$, first fix labels $x \sim \unif([k]^n)$. Then for each $S = \{s_1 < \cdots < s_p\} \in \binom{[n]}{p}$, draw $Y_S \sim \mu_{x_{s_1},\dots,x_{s_p}}$ independently and add $S$ to $E$ if $Y_S = 1$.
\end{enumerate}
\end{definition}

Although we show a rigorous end-to-end analysis of a quartic speedup only for this model, we conjecture that our algorithm produces a quartic speedup for more general SBM-like models. One type of generalized SBM (GSBM) that we study later (\Cref{def:pberngsbm}) replaces the function $f$ above by any $f:[k]^p\to \R$ such that, for any $p_* \geq 3$,
\begin{align}
    \E_{a \sim [k]^p}[f(a)] = 0, \qquad \E_{a_{r+1},\dots,r_p\in[k]}\left[f(a_1,\dots,a_p)\,|\,a_1,\dots,a_r\right] = 0 \,\forall\, r < p_*, \, a_1,\dots,a_r \in [k].
\end{align}

Even more generally, we expect that our quantum algorithm produces a quartic speedup for any GSBM with a \emph{marginal order} of $p_* \geq 3$, which we introduce in \Cref{def:morder}. These models can be shown to have a statistical-computational gap with precisely the same SNR-computation tradeoff as a $p_*$-order Tensor PCA problem (see \Cref{thm:kun} due to~\cite{kunisky2024low}). For any of these models, the Kikuchi algorithm takes a fairly generic form and can be readily applied. While we show its correctness only for specific models, we anticipate that more generic proofs can address these more general models.

The $p$-marginal HSBM approaches the $2$-marginal HSBM as the number of communities increases: in Lemma~\ref{lem:fprop}, we show that
\begin{align}
    f(a) \to 1_{a_1 = \cdots = a_{p_*}} \quad \mathrm{as} \quad k\to\infty,
\end{align}
i.e., the model assigns hyperedges within communities with probability $\theta_0 + \epsilon$ and between communities with probability $\theta_0$. Crucially, unlike the standard definition of HSBM, our model captures higher-order interactions of the underlying distribution. In many real-world contexts~\cite{momeni2017lotka,sanchez2019high,kim2006relating,kiel2011structural}, pairwise marginals are known to be insufficient to capture community structure.

The task of community detection that our algorithm solves is to distinguish samples from the planted distribution $\PP$ from samples from the null distribution $\QQ$. This is formalized as a strong detection task as follows.

\begin{problem}[Hypergraph community detection]
\label{prob}
    Let $p, k, n \geq 2$ and let $\QQ$ and $\PP$ be specified by a $p$-marginal HSBM. An algorithm that takes as input a degree-$p$ hypergraph $Y$ on $n$ vertices and outputs a bit $r(Y) \in \{0,1\}$ is said to solve the hypergraph community detection problem if 
    \begin{equation} 
    \Pr_{Y \sim \QQ}\left[ r(Y) = 1 \right] = 1- o(1) 
    \quad \text{ and } \quad 
    \Pr_{Y \sim \PP}\left[ r(Y) = 0\right] = 1 - o(1).
    \end{equation}
\end{problem}

For the task of hypergraph community detection, we show lower bounds that rule out \emph{low coordinate degree functions} (LCDF), which include low-degree polynomials. This is generically the best evidence one can hope for in average-case hardness.
\begin{theorem}[Lower bound on hypergraph community detection, informal]
    No function with coordinate degree $\ell$ can solve hypergraph community detection on a $p$-marginal HSBM with $p > 2$ if
    \begin{align}
        \mathrm{SNR} := \frac{\epsilon}{\sqrt{\theta_0(1-\theta_0)}} \lesssim \ell^{1/2-p/4}n^{-p/4}.
    \end{align}
\end{theorem}
A more formal statement is available in \Cref{thm:ldm_lowerbounds}; we introduce LCDF and prove this lower bound in \Cref{sec:low}.

We show a spectral classical algorithm based on the Kikuchi method that matches this lower bound (up to a logarithmic factor). Briefly, the algorithm works by constructing a matrix (the Kikuchi matrix) whose spectral norm either exceeds a threshold value if a hypergraph is sampled from $\PP$ or falls below that value if the hypergraph is sampled from $\QQ$.
\begin{theorem}[Classical upper bound on hypergraph community detection, informal]
    The $\ell$th order Kikuchi method solves hypergraph community detection for a $p$-marginal HSBM with even $p > 2$ in time $O(n^\ell)$ and space $O(n^\ell)$ if
    \begin{align}
        \mathrm{SNR} := \frac{\epsilon}{\sqrt{\theta_0(1-\theta_0)}} \gtrsim \ell^{1/2-p/4}n^{-p/4} \sqrt{\log n}.
    \end{align}
\end{theorem}
See \Cref{thm:kikuchi_bosonic} for a formal statement; we prove this theorem in \Cref{sec:classical_support}. The $\sqrt{\log n}$ factor, which also appears in previous algorithms using the Kikuchi method~\cite{WAM19,hastings2020classical, schmidhuber2025quartic}, is believed to be loose; indeed, it has already been removed for Tensor PCA \cite{kothari2025smoothtradeofftensorpca} and we expect it can be removed for $p$XORSAT and community detection as well, cf. \Cref{remark:remove_log}. Similarly, the constraint that $p$ is even is likely unnecessary~\cite{WAM19}.

We then show a quantum algorithm based on the guided Hamiltonian problem: We construct a guiding state that can be efficiently prepared and has improved overlap with the leading eigenspace of the Kikuchi matrix. Quantum Phase Estimation then prepares the eigenstate, and measurement of the corresponding eigenvalue can be used to certify its spectral norm and decide if a hypergraph was sampled from the planted (i.e., clustered) or null model. We show that this quantized Kikuchi method achieves a quartic speedup over the classical Kikuchi method, as summarized informally below. 
\begin{theorem}[Quantum upper bound on hypergraph community detection, informal]
    Given an $\ell$th order Kikuchi method that solves hypergraph community detection for a $p$-marginal HSBM instance with even $p > 2$ and any constant number of communities $k \geq 2$, there is an explicit quantum algorithm that solves the same problem on the same HSBM in time $n^{\ell/4}\cdot \Tilde O(n^p)$ up to negligible factors, using $\Tilde O(\ell \log(n))$ qubits and $\Tilde O(\ell n^{p})$ classical space.
\end{theorem}
Compared to the classical cost of $O(n^\ell)$, this algorithm achieves a quartic speedup and a super-polynomial space advantage. We formally state the result in \Cref{thm:kikuchi_q} and prove it in \Cref{sec:qalg}. We also anticipate that, analogous to the Tensor PCA case~\cite{hastings2020classical}, the eigenstate itself encodes the community labels. Performing state tomography and a rounding procedure should thus recover the community labels. 

\begin{remark}[Generality of quantum algorithm]
    Although our proofs for the quartic speedup use the precise form of $f$ given in \Cref{def:pbernsbm}, the quantum algorithm itself does not use any information about $f$. Hence, we expect exactly the same algorithm to apply to any generalized SBM with marginal order $p_* \geq 3$.
\end{remark}

\subsection{Organization of the paper}
The remainder of this paper is organized as follows. We give a detailed description of our classical Kikuchi upper bounds in \Cref{sec:Classical}, and of our low coordinate degree lower bounds in \Cref{sec:lower_bounds_lcdf}. \Cref{sec:Quantum} summarizes our quantum upper algorithm for hypergraph community detection. The corresponding technical proofs for the classical Kikuchi method are presented in \Cref{sec:classical_support} and \Cref{sec:low}, respectively, while the technical proofs related to our quantum algorithm are provided in \Cref{sec:qalg}.

\section{The Kikuchi method for \texorpdfstring{$p$}\ -marginal HSBM}\label{sec:Classical}
The $p$-marginal HSBM we study in this work belongs to the family of higher-order planted inference problems, which generally concern the detection of a signal or planted structure hidden in random noise. The development of efficient algorithms for these problems has seen intensive study for over two decades (see, e.g.,~\cite{Fei02,GK01,GJ02,CGL07,CCF10,AOW15,BM22,DT23,AR01,Sch08,OW14,MW16,KMOW17}), and is typically characterized by a Signal-to-Noise ratio $\mathrm{SNR}$. 

When the SNR is small enough, the problem is statistically impossible, but once the SNR exceeds a certain threshold, the two distributions can be distinguished, although not necessarily in polynomial time. For planted inference problems that exhibit a statistical-computational gap, that is, a regime where the problem is information-theoretically solvable but computationally hard, the Kikuchi hierarchy \cite{WAM19} is generally expected to be the the fastest and simplest algorithm in this regime. However, on a technically rigorous level the algorithmic Kikuchi method was as of now only shown to apply to Tensor PCA and $p$XORSAT~\cite{WAM19}, which are technically similar: $p$XORSAT can be recast as a variant of the spiked tensor problem with sparse Rademacher observations instead of Gaussian ones. Since then, all further applications of the Kikuchi algorithm have been restricted to directly related settings~\cite{guruswami2022algorithms,hsieh2023simple,alrabiah2023near,hsieh2024small}. 

In this work, we establish that the Kikuchi hierarchy extends to other planted inference problems by finding a new application in community detection. We prove matching lower bounds against low coordinate degree functions. Our results suggest that the classical and quantum Kikuchi algorithm is more widely applicable and may encompass more practically relevant applications than previously thought. At a technical level, our proofs require recent advancements in lower bounds~\cite{kunisky2024low} and more involved combinatorial arguments in our upper bounds due to the increased complexity of the problem. We summarize in this section the Kikuchi upper bound and show matching classical lower bounds (up to logarithmic factors) that rule out algorithms including low-degree polynomials.

\subsection{Kikuchi method}
The Kikuchi method, introduced by \cite{WAM19} and independently discovered by \cite{hastings2020classical}, is a general technique for reducing a degree-$p$ optimization problem to a degree-$2$ optimization problem. This is desirable because degree-$2$ problems may be modeled with matrices, allowing linear algebraic methods to be used.

For the $p$-marginal HSBM$(n,k,\theta_0,\eps)$, the relevant SNR is given by the ratio of the bias $\eps$ and the standard deviation $\sqrt{\theta_0(1-\theta_0)}$ of the random edge distribution. We show that the $p$-marginal HSBM$(n,k,\theta_0,\eps)$ exhibits a statistical-computational gap: It is solvable in polynomial time for \begin{equation}
    \frac{\eps}{\sqrt{\theta_0(1-\theta_0)}} \gg n^{-p/4},
\end{equation} and statistically impossible for \begin{equation}
    \frac{\eps}{\sqrt{\theta_0(1-\theta_0)}} \ll n^{(1-p)/2}.
\end{equation}
In between, the best classical algorithm is expected to be the Kikuchi method. \begin{center}
    \textbf{ Kikuchi method}
\end{center}
\noindent\hrulefill \\
\noindent\textbf{Input:} A $p$-marginal HSBM$(n,k,\theta_0,\eps)$ instance $\mathbf{Y}$.  \\
\noindent\textbf{Preprocessing:} Choose a sufficiently large $\ell$ and a threshold $\tau$ in accordance with a \emph{Kikuchi theorem} (see \Cref{thm:kikuchi_bosonic} for an example) at SNR $\beta = \eps / \sqrt{\theta_0(1-\theta_0)}$. That is, the Kikuchi theorem guarantees that if $\bY \sim \QQ$, then whp the spectral norm  $\norm{K_\ell}$ of the $\ell$-th order Kikuchi matrix is less than $\tau$; for $\bY \sim \PP$, it is larger than $\tau$. \\
\textbf{Classical algorithm:}
Construct the $\ell$-th order Kikuchi matrix and estimate its largest eigenvalue, for instance using the Power Method. \\
\textbf{Output:} If the largest eigenvalue is above $\tau$, return ``Planted''. Otherwise, return ``Random''.

\noindent\hrulefill

We describe the setting more formally. Let $\bY$ be a hypergraph drawn from the $p$-marginal HSBM$(n,k,\eps,\theta_0)$. Recall that for any hyperedge $S \in \binom{[n]}{p}$, the indicator function of the hyperedge is denoted $Y_S$. Define \begin{equation}
    A_S = \frac{Y_S - \theta_0}{\sqrt{\theta_0(1-\theta_0)}} \quad \text{ and } \quad \beta = \frac{\eps}{\sqrt{\theta_0(1-\theta_0)}}.
\end{equation} By construction, we have \begin{equation}
   \E_\QQ [A_S] = 0, \qquad \E_\QQ [A^2_S] = 1.
\end{equation} 
The original definition of the Kikuchi matrix due to \cite{WAM19} is as follows: For $S, V \in \binom{[n]}{\ell}$, define the symmetric difference $S\Delta V = S\cup V -S\cap V.$ The $\binom{n}{\ell} \times \binom{n}{\ell}$ Kikuchi matrix is defined entry-wise as
\begin{align}\label{eq:kik}
    \cK^{\mathrm{WAM}}_{S,V} = \begin{cases}
        A_{S \triangle V} & \text{if } |S \triangle V| = p\\
        0 & \mathrm{else}.
    \end{cases}
\end{align}   
We analyze our algorithm in the framework of a slight modification of the Kikuchi matrix, defined on $\ell$-tuples (that keep track of the ordering of indices) instead of $\ell$-sets. This matrix was originally studied by Hastings \cite{hastings2020classical} and achieves the same detection bounds as \cite{WAM19}.  We call this the ``bosonic Kikuchi matrix'', following the description in \cite{hastings2020classical}. The Kikuchi matrix corresponds to a symmetrized version of the bosonic Kikuchi matrix (cf. Appendix B of \cite{schmidhuber2025quartic}); this difference allows us to greatly simplify the analysis of our quantum algorithm. Throughout the following sections, we take $\ell = \lambda p$ for $\lambda \in \mathbb{N}$.
\begin{definition}
    Denote $\calT_n(\ell)$ the set of $\ell$-tuples $S = (i_1,\dots, i_\ell) \in [n]^\ell$ with no repeated entries. This set has cardinality $|\calT_n(\ell)| = n!/(n-\ell)!$.
\end{definition}
\begin{definition}
    Let $S = (i_1,\dots, i_\ell), V = (j_1,\dots, j_\ell) \in \calT_n(\ell)$ and define $$d = 2 \cdot | \{a \in [\ell] : i_a \not = j_a\}| $$ to count the number of disagreements. We denote by $S \ominus V$ the $d$-tuple that first lists the disagreeing entries of $S$ (in increasing index order) followed by those of $V$ (same order). 
\end{definition}
\begin{definition}[Bosonic Kikuchi matrix]
    For even $p \geq 2$, let $T$ be a symmetric $p$-tensor indexed by $(\mu_1, \dots, \mu_p) \in [n]^p$. The $\ell$-th order bosonic Kikuchi matrix $\cK =               \cK(T)$ is defined on $\calT_n(\ell)$ entry-wise by \begin{equation}
        \cK_{S,V} = \begin{cases}
            T_{\mu_1, \dots, \mu_p} \qquad \text{ if } (\mu_1, \dots, \mu_p) = S \ominus V, \\
            0 \qquad \qquad \ \ \ \text{ otherwise.}
        \end{cases}
    \end{equation}
\end{definition}
The dimension of the bosonic Kikuchi matrix differs from the dimension of the standard Kikuchi matrix \cref{eq:kik} only by a factor of $\ell!$, which is negligible compare to the overall dimension $O(n^\ell)$. Throughout this work, we assume our parameters satisfy conditions 
\begin{align}\label{eq:paramconds_sec}
  n^{-p/4} \gtrsim \frac{\epsilon}{\sqrt{\theta_0(1-\theta_0)}} \gtrsim n^{1/2-p/2}, \quad \theta_0 \gtrsim n^{1-p}, \quad \ell = o(\sqrt n), \quad \ell = \omega(1), \quad p \text{ even}
\end{align}
where $\gtrsim$ indicates asymptotic inequalities for sufficiently large $n$. The first condition simply states that the problem is neither in the information-theoretically impossible regime, nor in the computationally trivial regime; combined with the second condition, it also implies that $\epsilon/\theta_0 \leq 1$.
The following theorem, proven in \Cref{sec:classical_support}, establishes that the largest eigenvalue of $\cK$ distinguishes a hypergraph with community structure from one without. 
\begin{theorem}
Consider the $p$-marginal HSBM$(n,k,\theta_0,\eps)$ with $p$ even. Let $\ell \in[p / 2, n-p / 2], k \geq 2$, and $\epsilon, \theta_0$ in accordance with \eqref{eq:paramconds_sec}. Then for all
\begin{equation}
    \beta := \frac{\epsilon}{\sqrt{\theta_0(1-\theta_0)}} \geq \frac{3\sqrt{6}}{C_{k,p}} \ell^{1/2 - p/4} n^{-p/4}\sqrt{\log(n)},
\end{equation} where $C_{k,p}$ is a constant that depends only on $p$ and $k$, the $\ell$-th order bosonic Kikuchi matrix $\cK$ satisfies 
\begin{equation}
\Pr_\QQ\left[\lambda_{\max }(\cK) \geq \frac{2}{3}\tau \right] \vee \Pr_\PP\left[\lambda_{\max }(\cK) \leq \frac{4}{3}\tau \right] = o(1),\end{equation} where $\tau = \frac{1}{2}C_{k,p} \beta n^{p/2}\ell^{p/2}$. Hence, estimating the largest eigenvalue of the Kikuchi matrix solves hypergraph community detection.
\label{thm:kikuchi_bosonic}
\end{theorem}
\begin{remark}[SNR-computation tradeoff]
   For $p > 2$, the detection threshold smoothly interpolates between the computational threshold $\beta = \Tilde \Omega(n^{-p/4})$ at $\ell = O(1)$ and the information-theoretic threshold $\beta = \Tilde\Omega(n^{(1-p)/2})$ at $\ell = O(n)$. 
\end{remark}

\begin{remark}[Cost of Kikuchi algorithm]
    Generically, since $\cK$ is a matrix of dimension $O(n^\ell)$, the Kikuchi method corresponds to a classical algorithm that costs time and space $\tilde O(n^\ell)$. In comparison, our quantum algorithm will use  $\tilde O(\ell \log n)$ qubits and $\tilde O(\ell n^p)$ classical bits, as well as achieving a quartic speedup in time.
\end{remark}

\begin{remark}[Threshold choice]
    The multiplicative factor of $2$ between the random and planted threshold in \Cref{thm:kikuchi_bosonic} is chosen arbitrarily for ease of exposition; in practice, any factor $1 + \Omega(1/\poly(n))$ suffices, resulting in a slightly smaller choice of $\ell$ classically and quantumly. 
\end{remark}

\section{Lower bounds against low coordinate degree functions}
\label{sec:lower_bounds_lcdf}
Here, we show that the Kikuchi algorithm is classically optimal for a family of models including the $p$-marginal HSBM under standard conjectures in average-case hardness (\Cref{thm:ldm_lowerbounds}). Our lower bounds will address a larger class of models that share the form of the $p$-marginal HSBM definition but with looser conditions on $f$.

Due to the difficulty of showing lower bounds against all possible classical algorithms, the evidence in planted / average-case settings is typically restricted to ruling out a large class of classical algorithms. One common class is \emph{low-degree polynomial} (LDP) algorithms: these solve a detection problem by computing a low-degree polynomial of the input data $x$ whose whose value is different under the planted $x \sim \PP$ and the null $x \sim \QQ$ distributions. LDP lower bounds are typically shown through the low-degree likelihood ratio.

Recently, techniques for lower bounds against low-degree polynomial algorithms have been generalized to a larger class of classical algorithms: \emph{low coordinate degree functions} (LCDF)~\cite{kunisky2025low}. While LDP are linear combinations of low-degree monomials, LCDF are linear combinations of arbitrary functions of entries in a small number of coordinates. We review the definitions of LCDF, strong separation, and strong detection.

\begin{definition}[Strong separation]
    Consider a sequence of pairs of probability measures $\PP_n, \QQ_n$ over measurable spaces $\Omega_n$. We say that functions $f_n:\Omega_n\to\R$ achieve \emph{strong separation} if
    \begin{align}
        \E_{Y \sim \PP_n} f_n(Y) - \E_{Y \sim \QQ_n} f_n(Y) = \omega\lr{\sqrt{\Var_{Y\sim\QQ_n} f_n(Y)} + \sqrt{\Var_{Y\sim\PP_n} f_n(Y)}}
    \end{align}
    as $n \to\infty$.
\end{definition}

\begin{definition}[Strong detection]
    Consider a sequence of probability measures $\PP_n, \QQ_n$ over $\Omega_n$. We say that functions $r_n:\Omega_n\to\{0, 1\}$ achieve \emph{strong detection} if
    \begin{align}
        \lim_{n\to\infty} \PP_n[r_n(y) = 0] = \lim_{n\to\infty} \QQ_n[r_n(y)=1] = 0,
    \end{align}
    that is, if the sequence of hypothesis tests $r_n$ have both Type I and Type II error probabilities tending to zero.
\end{definition}

Strong separation has an operational implication in hypothesis testing: given a strongly separating family of functions, one obtains hypothesis tests with $o(1)$ Type I and II errors. The converse is also true; hence, strong separation by a family of functions is equivalent to strong detection. We also note that, as discussed in~\cite{wein2025computational}, the task of detection rather than refutation is typically best lower bounded by ruling out algorithms that are low (coordinate) degree rather than SoS; for high-dimensional statistical problems, such as ours, either approach is expected to perform equally well.

Our proofs address the family of functions given by LCDF. For a product measure $\QQ$ on $\Omega^N$ (such as in our $p$-marginal HSBM, \Cref{def:pbernsbm}), we write $y_T \in \Omega^T$ for the restriction of $y$ to the coordinates in $T \subseteq [N]$. We define subspaces of $L^2(\QQ)$:
\begin{align}
    V_T &= \{f \in L^2(\QQ)\,:\,f(y)\text{ depends only on }y_T\}\\
    V_{\leq D} &= \sum_{T\subseteq [N] \,:\, |T| \leq D} V_T.
\end{align}
The space $V_{\leq D}$ is the space of low degree coordinate functions, formalized as follows.

\begin{definition}[Coordinate degree]
    For $f \in L^2(\QQ)$, the \emph{coordinate degree} of $f$ is
    \begin{align}
        \mathrm{cdeg}(f) = \min\{D\,:\, f \in V_{\leq D}\}.
    \end{align}
\end{definition}

\begin{remark}
    When $\Omega \subseteq \R$, functions of coordinate degree at most $D$ include polynomials of degree at most $D$.
\end{remark}

We can now state our LCDF lower bound for $p$-marginal HSBM formally; it is proven in \Cref{sec:low}.

\begin{theorem}[Lower bound on hypergraph community detection]
    Let $\PP_n$ and $\QQ_n$ be the planted and null distributions of a $p$-marginal HSBM$(n,k,\theta_0,\eps)$ on $n$ vertices with $p > 2$. No sequence of functions of coordinate degree at most $\ell$ can strongly separate $\QQ_n$ from $\PP_n$ for
    \begin{align}
        \beta := \frac{\epsilon}{\sqrt{\theta_0(1-\theta_0)}} < C_{k,p} \ell^{1/2-p/4}n^{-p/4}
    \end{align}
    for some constant $C_{k,p}$ independent of $n$.
    \label{thm:ldm_lowerbounds}
\end{theorem}
\Cref{thm:kikuchi_bosonic,thm:ldm_lowerbounds} together establish that the Kikuchi method is optimal, up to the logarithmic factor $\sqrt{\log(n)}$ which is negligible whenever $\ell \gg \log(n).$
We note that this logarithmic gap between the lower and upper bounds is a common artifact introduced by the Matrix-Chernoff bound appearing in the analysis of the Kikuchi method, and also appears in previous work \cite{WAM19,hastings2020classical, schmidhuber2025quartic}. Just like these prior works, we conjecture that the logarithmic gap can be closed by a tighter analysis, such as the one recently carried out for Tensor PCA \cite{kothari2025smoothtradeofftensorpca} (cf. \Cref{remark:remove_log}).
\section{Quantum algorithm for \texorpdfstring{$p$}\ -marginal HSBM}
\label{sec:Quantum}

\subsection{Quantum algorithm}

Our quantum algorithm for community detection is a quantized version of the Kikuchi method; hence, much of the required technical work is achieved by showing the theorems claimed in the previous section. Just like the classical algorithm, it distinguishes a hypergraph with community structure from a random hypergraph by estimating the largest eigenvalue of the $\ell$th order Kikuchi matrix for a suitable choice of $\ell$. However, our quantum algorithm does so using $\Tilde{O}(n^{\ell/4} \cdot \poly(n))$ quantum gates and only $\Tilde{O}(\ell \log(n))$ qubits. On the contrary, estimating this eigenvalue classically, say, using the Power Method, requires time at least $\Omega(n^\ell)$. (Note that simply writing down one vector in the dimension of $\calK_\ell$ takes $\Omega(n^\ell)$ time and space.) Hence our quantum algorithm constitutes a (nearly) quartic speedup in time as well as superpolynomial savings in space over the best know classical algorithm. 
 
\begin{center}
    \textbf{ Quantized Kikuchi method}
\end{center}
\noindent\hrulefill\\
\noindent\textbf{Input:} A $p$-marginal HSBM$(n,k,\theta_0,\eps)$ instance $\mathbf{Y}$.  \\
\noindent\textbf{Preprocessing:} As in the classical Kikuchi algorithm, choose a sufficiently large $\ell$ and a threshold $\tau$ in accordance with a classical ``Kikuchi Theorem'' at SNR $\beta$, such as \Cref{thm:kikuchi_bosonic}.\\
\textbf{Quantum algorithm:}
Encode the following in Amplitude Amplification and repeat $O(n^{\ell/4})$ times:
\begin{itemize}[left = 0.2cm]
        \item Prepare $\ell/p$ unentangled copies of a ``small'' guiding state $\ket{\phi}$. Symmetrize the resulting state to obtain a guiding state $\ket{\Phi}$. 
        \item Perform Quantum Phase Estimation with the sparse Hamiltonian $K_\ell$ on the initial state $\ket{\Phi}$.
        \item Measure the eigenvalue register and record whether an eigenvalue above the threshold $\tau$ was sampled.
\end{itemize}
\textbf{Output:} If during any of the repetitions, an eigenvalue above $\tau$ was found, return ``Planted''. Otherwise, return ``Random''. 

\noindent\hrulefill

The quartic quantum speedup is a combination of two quadratic speedups; one is standard and due to Amplitude Amplification \cite{brassard2000quantum}. The other is due to the construction of an efficient guiding state that has improved overlap with the leading eigenspace of the Kikuchi matrix in the planted case. We now discuss this guiding state in some detail. 

Recall that the Kikuchi method is spectral: $\norm{\cK}$ exceeds $\tau$ when an instance comes from the planted distribution and falls beneath it when the instance comes from the null distribution. A natural state that certifies this spectral norm via the variational lower bound $\norm{\cK} \geq \bra{v}\cK\ket{v}$ is (tensor products of) the state
\begin{align}
    \ket{v}_S \propto f(x_S), \qquad \left(\text{ for } S \in \binom{[n]}{p}\right).
\end{align}
For details, see \Cref{def:certificate}. However, preparing this state requires knowledge of the community labels $x$. With only the hypergraph available, we cannot prepare $\ket{v}$ as a guiding state. To get around this issue, recall that the hyperedge on set $S$ of a $p$-marginal HSBM is placed if $Y_S=1$ for random variable $Y_S \sim \bern(\theta_0 + \epsilon f(x))$. Hence, preparing a state proportional to 
\begin{align}
    \ket{u}_S \propto Y_S - \theta_0
\end{align}
serves as an approximation to $\ket{v}$. This motivates our guiding state, which we define formally in \Cref{def:guide_vector}. Informally, we first show the following result, a formal version of which is given in \Cref{lem:guideoverlap}.
\begin{lemma}[Guiding state overlap with certificate state, informal]\label{lem:momenttail_informal}
    Let $\ket{u}, \ket{v}$ be the unnormalized certificate state and guiding states of \Cref{def:certificate,def:guide_vector}.
    Then
    \begin{align}
        \mathrm{Pr}_\PP \left[\frac{\bra{u}\ket{v}}{\norm{u}\norm{v}} \leq \Tilde{O}\left(n^{-\ell/4} (\ell!)^{1/2p-1/4} \right)\right] &= o(1).
    \end{align}
\end{lemma}
\begin{remark}
    When $\ell = o(n^\eps)$ for any $\eps > 0$, \Cref{lem:momenttail_informal} shows that guiding state has sufficient overlap to achieve a quartic quantum speedup up to logarithmic factors.
\end{remark}
A crucial difficulty for all existing quantum Kikuchi algorithms is to show that the leading eigenspace of the Kikuchi matrix (which can be interpreted as a ``noisy version'' of the certificates) retains its good overlap with the guiding state. Although there is no clear reason to expect that this perturbation significantly degrades the overlap, proving the rigorously is technically challenging. Similar obstacles were faced (and overcome) by \cite{hastings2020classical} and \cite{schmidhuber2025quartic}. The work of \cite{hastings2020classical} uses a somewhat intricate strategy to deal with it, drawing from special properties of Gaussian random variables and low-degree polynomials (e.g., the Carbery--Wright theorem on anticoncentration). On the other hand, \cite{schmidhuber2025quartic} first enforces independence between the guiding state and the Kikuchi matrix by building them out of separate batches of the input, and subsequently use a second moment method. We overcome this obstacle using a similar overall strategy as this second approach.
The result is formally stated in \Cref{thm:final_kikuchi_bosonic}, and shows that the improved overlap is maintained up to factors that are negligible compared to the overall scaling of $n^{O(\ell)}$.

Our main result applies this guiding state to obtain a quartic quantum speedup for the $p$-marginal HSBM problem over the best known classical algorithm via the guided Hamiltonian problem. It applies in the same setting as the classical Kikuchi algorithm described in \Cref{sec:Classical}. We show how to efficiently prepare the guiding state $\ket{\Tilde{u}}$ by preparing a tensor product of ``small guiding states'' and subsequently projecting into a collision-free subspace (cf. \Cref{sec:symmetrization}).
By running quantum phase estimation starting from the guiding state, our quantum algorithm prepares a high-energy eigenstate and measures its energy; repeating this step $\Tilde{O}(n^{\ell/4})$ times (combined with amplitude amplification) solves strong detection on the hypergraph stochastic block model. The cost of phase estimation depends on the sparsity of the Kikuchi matrix and the cost of preparing the sparse oracles; since both of these are at most $n^{p/2}$ (see \Cref{sec:space_adv}), we obtain the following key result, which we prove in \Cref{sec:putting_together}.
\begin{theorem}
    Let $p > 2$ be even and $k\geq 2$. Let $\bY \sim \mathrm{HSBM}(k,n,\eps,\theta_0)$ be a $p$-marginal HSBM instance in accordance with \cref{eq:paramconds_sec}. Let $\ell \in[p / 2, n-p / 2]$ be chosen such that the classical Kikuchi method specified by \Cref{thm:kikuchi_bosonic} solves the community detection problem for $\bY$ in time $\Tilde O(n^{\ell})$ at SNR $\beta = \eps/\sqrt{\theta_0(1-\theta_0)}$. Then the quantum algorithm solves the same problem at the same SNR using $$\Tilde O\left(n^{\ell/4}\cdot n^p \cdot \ell^{\frac{\ell}{4}-\frac{\ell}{2p}}  \cdot \log(n)^{\ell/2p} \cdot \exp(O(\ell))\right)$$ quantum gates, $\Tilde{O}(\ell \log(n))$ qubits, and classical space $\tilde O(\ell n^{p})$.
    \label{thm:kikuchi_q}
\end{theorem}

\begin{remark}
    For constant or slowly growing  $\ell$ (specifically, $\ell = o(n^\eps)$ for any $\eps > 0$) this gives a (nearly) quartic speedup in time and a superpolynomial reduction in space. For general parameters satisfying \cref{eq:paramconds_sec}, the speedup interpolates between quartic and super-quadratic. For example, at $p=4$ and the maximal value of $\ell \lesssim \sqrt{n}$, the degree of the speedup approaches $\frac{16}{5}$, which is still super-cubic.
\end{remark}
\subsection{Quantum space advantage}
\label{sec:space_adv}
We end with final remarks on the space advantage of the quantum Kikuchi algorithm, which is already well established in \cite{hastings2020classical} and \cite{schmidhuber2023complexity}.
Our model has $O(\theta_0 + \epsilon)\binom{n}{p}\sim (\theta_0 + \epsilon) n^p$ hyperedges. The sparsity of the Kikuchi method is at most $O(n^{p/2}$. Hence, implementing the sparse oracles (Theorem 71 in~\cite{schmidhuber2025quartic}) uses $O(n^p \ell \log n)$ gates and $O(\ell \log n)$ qubits; the overall number of qubits is $\Tilde O(\ell \log n)$. If we account for the classical space requirement of storing the input ($O(n^p)$) this is a superpolynomial space advantage, as claimed throughout this paper. If we only count the quantum space requirements of working qubits required by our algorithm, this is an exponential improvement over the classical space requirement of the classical Kikuchi method.

\section{Conclusion}
Our work identifies a new application for the classical Kikuchi algorithm and shows that the quantized algorithm maintains a quartic speedup. This suggests that the speedup identified in the original Kikuchi quantum algorithm of~\cite{hastings2020classical} is preserved in a larger class of inference problems with a statistical-computational gap. We remark on several related questions our work raises with regards to this larger class of problems.
\begin{itemize}
    \item \emph{Odd order of hypergraph.} Our result holds for hypergraphs with even order. It has been long understood that the analysis of the Kikuchi method for even order is simpler than odd order~\cite{WAM19,guruswami2022algorithms}, but we expect the quantum speedup to be preserved for HSBMs with odd order.
    \item \emph{Applications beyond community detection.} Although we analyze the particular setting of community detection, we note that the quantum algorithm we describe may be generically applicable in random hypergraph settings. Given that marginal order provides a sufficient condition for a statistical-computational gap, a natural conjecture is that the Kikuchi quantum speedup holds for all such generalized SBMs.
    \item \emph{Recovery of community labels.} Beyond detection, practical applications are often interested in recovering the community structure. In the quantum Kikuchi algorithm for Tensor PCA, the ground state contains sufficient information to recover the planted spike via state tomography and a rounding procedure. We leave open the analogous question of recovering community labels here.
\end{itemize}

\bibliographystyle{alpha}
\bibliography{references}

\newpage
\appendix
\section{Correctness and efficiency of the classical Kikuchi algorithm}
\label{sec:classical_support}
In this section, we formally describe the classical algorithm based on the Kikuchi hierarchy, and prove the detection threshold in \Cref{thm:kikuchi_bosonic}. The upper bound in the null case follows from a simple Matrix Chernoff bound. Proving the lower bound in the planted case is straight-forward in previous work \cite{WAM19,hastings2020classical,schmidhuber2025quartic} for Tensor PCA and $p$XORSAT, because a simple certificate with product structure suffices. However, the natural certificate in our setting has entries proportional to $f(x)$, which do not exhibit a product structure. As a result, our proof of \Cref{thm:kikuchi_bosonic} uses more complex combinatorial arguments. 

Recall the setup from \Cref{sec:technical_summary}. We draw random variables
\begin{align}
    Y_S \sim \bern(\theta_0 + \epsilon f(x_S))
\end{align}
conditioned on some $x \in [k]^n$, and we define
\begin{align}
    f(x_S) = \sum_{i=1}^k \prod_{j=1}^p \lr{1_{x_{S_j}=i} - \frac{1}{k}}.
\end{align}
We consider drawing $x \sim [k]^n$ uniformly at random and then drawing all $Y_S$. We use $\beta=\epsilon\sqrt{\frac{1}{\theta_0(1-\theta_0)}}$ to denote the Signal-to-Noise ratio (SNR).
The proof of \Cref{thm:kikuchi_bosonic} boils down to two statements: \begin{itemize}
    \item Show that in the planted case, there exists a \emph{certificate} vector $\ket{v}$ such that $\beta n^{p/2}\ell^{p/2}$ with high probability. 
    \item Show (using Matrix-Chernoff) that in the null case, no such vector exists.
\end{itemize}

\subsection{Kikuchi hierarchy}\label{sec:kik}
We repeat the definition of the Kikuchi matrix here.
\begin{definition}
    Denote $\calT_n(\ell)$ the set of $\ell$-tuples $S = (i_1,\dots, i_\ell) \in [n]^\ell$ with no repeated entries. This set has cardinality $|\calT_n(\ell)| = n!/(n-\ell)!$.
\end{definition}
\begin{definition}
    Let $S = (i_1,\dots, i_\ell), V = (j_1,\dots, j_\ell) \in \calT_n(\ell)$ and define $$d = 2 \cdot | \{a \in [\ell] : i_a \not = j_a\}| $$ to count the number of disagreements. We denote by $S \ominus V$ the $d$-tuple that first lists the disagreeing entries of $S$ (in increasing index order) followed by those of $V$ (same order). 
\end{definition}
\begin{definition}[Bosonic Kikuchi matrix]
    For even $p \geq 2$, let $T$ be a symmetric $p$-tensor indexed by $(\mu_1, \dots, \mu_p) \in [n]^p$. The $\ell$-th order bosonic Kikuchi matrix $\cK =               \cK(T)$ is defined on $\calT_n(\ell)$ entry-wise by \begin{equation}
        \cK_{S,V} = \begin{cases}
            T_{\mu_1, \dots, \mu_p} \qquad \text{ if } (\mu_1, \dots, \mu_p) = S \ominus V, \\
            0 \qquad \qquad \ \ \ \text{ otherwise.}
        \end{cases}
    \end{equation}
\end{definition}

\begin{definition}[Homogeneous subspace]\label{def:hom}
    Given subset $S \in \binom{[n]}{p}$, define the symmetrized basis vector $\ket{S} = \frac{1}{\sqrt{p!}} \sum_{\pi \in \Sym(p)} \ket{\pi(i_1,\dots,i_p)}$. Define the collision-free index set
    \begin{align}
        \cC_m = \left\{(S_1, \dots, S_m) \in \binom{[n]}{p}^m \,:\, S_a \cap S_b = \emptyset \, \forall\, a \neq b\right\}.
    \end{align}
    The homogeneous subspace projector $\Pi_m$ projects onto collision-free tuples
    \begin{align}
        \Pi_m = \sum_{(S_1,\dots,S_m) \in \cC_m} \ketbra{S_1} \otimes \cdots \otimes \ketbra{S_m}.
    \end{align}
\end{definition}

The following facts about the collision-free index set will be useful.
\begin{lemma}[Counting overlaps]\label{lem:coverlap}
    For $\lambda =  o(\sqrt n)$, the collision-free index set has size
    \begin{align}
        |\cC_\lambda| = \prod_{r=0}^{\lambda-1} \binom{n-rp}{p} = \binom{n}{p}^\lambda \lr{1 + O\lr{\frac{\lambda^2}{n}}}.
    \end{align}
    Moreover, the number $r(C, C')$ of coordinates contained in at least one set of $C \in \cC_\lambda$ and at least one set of $C' \in \cC_\lambda$ satisfies
    \begin{align}
        \#\{(C, C') \,:\, r(C, C') = r\} \leq |\cC_\lambda|^2 \binom{\lambda p}{r} \lr{\frac{\lambda p}{n}}^r.
    \end{align}
\end{lemma}
\begin{proof}
    We compute
    \begin{align}
        \#\{(C, C') \,:\, r(C, C') = r\} = |\cC_\lambda|^2 \pr{r(C, C')=r}
    \end{align}
    over uniformly sampled $C, C' \in \cC_\lambda$. The elements of $C$ are uniform over $\binom{[n]}{\lambda p}$; the probability $C'$ has $r$ overlapping elements is thus
    \begin{align}
        \frac{\binom{n-r}{\lambda p - r}}{\binom{n}{\lambda p}} \leq \lr{\frac{\lambda p}{n}}^r.
    \end{align}
    There are $\binom{\lambda p}{r}$ choices of where to place the overlapping elements and thus
    \begin{align}
        \pr{r(C, C')=r} \leq \binom{\lambda p}{r} \lr{\frac{\lambda p}{n}}^r.
    \end{align}
\end{proof}

\begin{definition}[Certificate vector]
\label{def:certificate}
    Let $\ket{v'}$ be the unnormalized vector given by
    \begin{align}
        \ket{v'} = \sum_{S \in \binom{[n]}{p}} f(x_S) \ket{S}.
    \end{align}
    The certificate vector $\ket{v}$ is the unnormalized vector $\Pi_\lambda \ket{v'}^{\otimes \ell/p}$.
\end{definition}

\subsection{Analysis in the random case}
\begin{lemma}\label{lem:bosonic_nullknorm}
The bosonic Kikuchi matrix of a $p$-marginal HSBM satisfies
\begin{equation}
    \Pr_\QQ \left[  \norm{\cK} \geq \sqrt{6 n^{p/2} \ell^{1+p/2} \log(n)} \right] \leq o(1).
\end{equation}
\end{lemma}
\begin{proof}
    Write the bosonic Kikuchi matrix as \begin{equation}
        \cK = \sum_{E \in \binom{[n]}{p}} A_E \cK^{(E)}, \qquad \cK_{S,V}^{(E)} = \begin{cases}
            1 \qquad \text{ if }   S \ominus V = (\mu_1, \dots, \mu_p)\textbf{ with } \{\mu_1, \dots, \mu_p\} = E, \\
            0 \qquad \text{ otherwise.}
        \end{cases}
    \end{equation} The $A_E = (Y_E-\theta_0)/\sqrt{\theta_0(1-\theta_0)}$ are independent random variables associated to the hyperedges of the $p$-marginal HSBM, with mean zero and variance one. $K_{S, V}^{(E)}=1$ if and only if $S$ and $V$ differ in exactly $p / 2$ positions and $S \ominus V$ has underlying set $E$ (order ignored). For fixed $S$ and $E$ with $|E \cap S|=p / 2$, the first $p / 2$ entries of $S \ominus V$ are determined by $S$; there are $(p / 2)!$  permutations for the $V$-half, hence $\left\|K^{(E)}\right\|_{\infty} \leq(p / 2)!$  and $R := \lambda_{\text{max}}(\cK^{(E)}) \leq (p/2)!$.
The Matrix Bernstein inequality \cite{Tro12} states that 
    \begin{equation}
\begin{aligned}
\operatorname{Pr}\left\{\lambda_{\max }\left(\calK \right) \geq t\right\} & \leq \frac{n!}{(n-\ell)!} \exp \left(-\frac{\sigma^2}{R^2} \cdot h\left(\frac{R t}{\sigma^2}\right)\right) \\
& \leq \frac{n!}{(n-\ell)!} \exp \left(\frac{-t^2}{\sigma^2+ Rt / 3}\right) \\
& \leq \begin{cases}\frac{n!}{(n-\ell)!}  \exp \left(-3 t^2 / 8 \sigma^2\right) & \text { for } t \leq \frac{\sigma^2}{R} ; \\
\frac{n!}{(n-\ell)!} \exp (-3 t / 8R) & \text { for } t \geq \frac{\sigma^2}{R}.\end{cases}
\end{aligned}
\end{equation}
Here, we have introduced the notation $\sigma^2$ for the ``variance'' \begin{equation}
    \sigma^2 = \norm{\sum_j \E_\QQ \left[\left(\cK^{(E)}\right)^2\right]}.
\end{equation} 
To compute this quantity, we note that \begin{equation}
\left(\cK^{(E)}\right)^2 \ket{S} = \begin{cases}
           (p/2)! \sum_{\pi \in \Pi_{S|E}} \ket{\pi(S)} \qquad \text{ if }   |S \cap E| = p/2, \\
            0 \qquad \text{ otherwise.}
        \end{cases}
    \end{equation}
Here, $\Pi_{S|E}$ denotes the set of permutations that reorder the $p / 2$ indices of $S$ lying in $E$ while fixing all other coordinates. 
Summing over all sets $E$ yields \begin{equation}
\sum_E \left(\cK^{(E)}\right)^2 \ket{S} = 
           \binom{n-\ell}{p/2}(p/2)!
           \sum_{P \in \binom{[\ell]}{p/2}}
           \sum_{\pi \in \Pi_{S|P}} \ket{\pi(S)} .
    \end{equation}
By a row sum bound, the spectral norm of the variance matrix is therefore \begin{equation}
    \sigma^2 \leq  \binom{n-\ell}{p/2}\binom{\ell}{p/2}(p/2)!^2 = (1-o(1)) \cdot n^{p/2}\ell^{p/2}.
\end{equation}
For the choice $t = \sqrt{6 \sigma^2 \ell \log(n)}$ we have $t = o( \sigma^2/R)$ for $p \geq 4$, and hence Matrix Bernstein implies, for large enough $n$, that \begin{equation}
     \Pr[\norm{\cK} \geq \sqrt{6 \sigma^2 \ell \log(n)}] \leq \frac{n!}{(n-\ell)!} \exp \left(-18 / 8 \cdot \ell \log n \right) =  o\lr{n^{-\ell}}.
\end{equation}
\end{proof}
\begin{remark}
    The constant $\sqrt{6}$ is chosen almost arbitrarily and can be improved by a tighter analysis. 
\end{remark}

\subsection{Analysis in the planted case}
We compute several necessary moments in terms of the quantity
\begin{align}
    \mu = \EE{f(x_S)^2} = k^{1-2p}\lr{(k-1)^p + (k-1)}.
\end{align}
We also compute moments in terms of the following quantity.
\begin{lemma}\label{lem:in_kikmean}
    Let $r \geq 1$ and $u \in [r]$; for $s_u \in \{0,1,\dots,p\}$, any sets $A_u, B_u \in \binom{[n]}{p}$ such that
    \begin{align}
        |A_u \cap A_v| = |B_u \cap B_v| = p\delta_{uv}, \qquad |A_u \cap B_u| = s_u, \qquad \sum_{u=1}^r |A_u \triangle B_u| = p
    \end{align}
    we have for $C = \sqcup_u A_u \triangle B_u \in \binom{[n]}{p}$ that
    \begin{align}
        \E_x \lr{\prod_{u=1}^r f(x_{A_u}) f(x_{B_u})} f(x_C) = k \prod_{u=1}^r g_{k,p}(s_u)
    \end{align}
    for
    \begin{align}
        g_{k,p}(s) = a^{2p-s} + 2(k-1)a^{p-s}b^p + (k-1)a^s b^{2(p-s)} + (k-1)(k-2)b^{2p-s}, \quad a = \frac{k-1}{k^2}, \quad b=-\frac{1}{k^2}.
    \end{align}
\end{lemma}
\begin{proof}
    We compute
    \begin{align}
        \EE{\lr{1_{x_t=i} - \frac{1}{k}}\lr{1_{x_t=j} - \frac{1}{k}}} = \begin{cases}
            a & i=j\\
            b & i \neq j
        \end{cases}
    \end{align}
    for $a, b$ defined in the lemma statement. Observe that each index $t$ appears in exactly two sets: either $t \in A_u \cap B_u$, or $t \in A_u \triangle B_u$ and exactly one of $A_u, B_u$. Fix some $u$ and $c \in [k]$. Then the $u$th block contributes
    \begin{align}
        h_{k,p}(s_u, c) &= \sum_{i,j \in [k]} \lr{\prod_{t \in A_u \cap B_u} \EE{1_{x_t=c} - \frac{1}{k}}} \lr{\prod_{t \in A_u \cap C} \EE{1_{x_t=c} - \frac{1}{k}}} \lr{\prod_{t \in B_u \cap C} \EE{1_{x_t=c} - \frac{1}{k}}}.
    \end{align}
    The choice $i=j=c$ appears once and contributes $a^{2p-s_u}$; the choice $i=c, j\neq c$ appears $k-1$ times and contributes $a^{p-s_u}b^p$ and similarly for $i \leftrightarrow j$; the choice $i=j \neq c$ appears $k-1$ times and contributes $a^s b^{2(p-s_u)}$; the choice where $i,j,c$ are all distinct appears $(k-1)(k-2)$ times and contributes $b^{2p-s_u}$. Hence, $h_{k,p}(s_u,c) = g_{k,p}(s_u)$ as defined in the lemma statement.
    By independence of the coordinates in $x$, we can rewrite
    \begin{align}
        \E_x \lr{\prod_{u=1}^r f(x_{A_u}) f(x_{B_u})} f(x_C) = \sum_{c=1}^k \prod_{u=1}^r h_{k,p}(s_u, c) = k \prod_{u=1}^r g_{k,p}(s_u).
    \end{align}
\end{proof}

\begin{lemma}[Expectation of certificate energy]\label{lem:kmean}
    Assuming $\ell = o(\sqrt n)$, we have
    \begin{align}
        \E_\PP \bra{v}\cK\ket{v} = \beta C_{k,p} \binom{n}{p}^\lambda \mu^\lambda n^{p/2} \lambda^{p/2}(1+o(1))
    \end{align}
    for constant $C_{k,p}$ independent of $\lambda, n, \beta$.
\end{lemma}
\begin{proof}
    For $S, T \in \cC_\lambda$, we track the indices at which blocks agree and disagree by
    \begin{align}
        \ag(S, T) = \{i \in [\lambda] \,:\, S_i = T_i\}, \qquad \dis(S, T) = [\lambda] \setminus \ag(S, T).
    \end{align}
    Note that $\ket{v}$ is only supported on $S$ such that for all $i \neq j$, $|S_i \cap S_j| = 0$.
    Recall that each set $S_i \in S$ is encoded as state $\ket{S_i} = \frac{1}{\sqrt{p!}} \sum_{\pi \in \Sym(p)} \ket{\pi(j_1,\dots,j_p)}$ for $S_i = \{j_1,\dots,j_p\}$. Since matrix elements of $\cK$ are indexed by an ordered tuple $(j_1,\dots,j_\ell) \in \cT_n(\ell)$, we have matrix element
    \begin{align}
        \bra{S}\cK\ket{T} = \frac{1}{(p!)^\lambda} \sum_{\substack{\pi_1,\dots,\pi_\lambda\\\sigma_1,\dots,\sigma_\lambda}} \cK_{(\pi_1(S_1)|| \cdots ||\pi_\lambda(S_\lambda)),(\sigma_1(T_1)|| \cdots ||\sigma_\lambda(T_\lambda))}.
    \end{align}
    Since the elements of $\cK$ are either zero or elements of a symmetric tensor, we simply obtain the following:
    \begin{align}
        \E_\PP \bra{v}\cK\ket{v} &= \E_\PP \sum_{\substack{(S_1,\dots,S_\lambda)\in \cC_\lambda\\ (T_1,\dots,T_\lambda)\in \cC_\lambda}} \lr{\prod_{i=1}^\lambda f(x_{S_i})f(x_{T_i})} \bra{S_1,\dots,S_\lambda}\cK\ket{T_1,\dots,T_\lambda}\\
        &= \beta \E_x \sum_{\substack{(S_1,\dots,S_\lambda)\in \cC_\lambda\\ (T_1,\dots,T_\lambda)\in \cC_\lambda\\ |S\triangle T|=p}} \lr{\prod_{i \in \ag(S,T)}f(x_{S_i})^2} \lr{\prod_{i\in\dis(S,T)} f(x_{S_i})f(x_{T_i})} f(x_{S\ominus T}).
    \end{align}
    To evaluate the sum, we sum over the number $r = |\dis(S,T)| \in \{1,\dots,p/2\}$ of disagreeing blocks; label the blocks in $S$ as $A_1, \dots, A_r$ and the blocks in $T$ as $B_1, \dots, B_r$. We also sum over the number of overlaps in each disagreeing block, $s_u = |A_u \cap B_u|$; i.e., over the set
    \begin{align}
        \cS_r = \left\{(s_1, \dots, s_r) \in \{0,\dots,p\}^r \,:\, \sum_{u=1}^r s_u = rp - \frac{p}{2}\right\},
    \end{align}
    where the constraint ensures that
    \begin{align}
        p = |S \ominus T| = \sum_{u=1}^r |A_u \triangle B_u| = \sum_{u=1}^r 2(p-s_u).
    \end{align}
    This gives for a particular $r$ and $s$ the term
    \begin{align}
        \E_x \lr{\prod_{i \in \ag} f(x_{T_i})^2} \lr{\prod_{u=1}^r f(x_{A_u}) f(x_{B_u})} f(x_C) &= \lr{\prod_{i \in \ag} \E_x f(x_{T_i})^2} \E_x \lr{\prod_{u=1}^r f(x_{A_u}) f(x_{B_u})} f(x_C) \\
        &= \mu^{\lambda-r} \psi_r(s),
    \end{align}
    where $C = \sqcup_u A_u \triangle B_u \in \binom{[n]}{p}$ and $\psi_r(s)$ is given by \Cref{lem:in_kikmean}. For some counting factor $N_\disj(r,s)$, we thus have
    \begin{align}
        \E_\PP \bra{v}\cK\ket{v} &= \beta \sum_{r=1}^{p/2} \sum_{s\in\cS_r} N_\disj(r,s) \mu^{\lambda-r}\psi_r(s).
    \end{align}
    We compute the counting factor step by step.
    \begin{itemize}
        \item There are $|\cC_\lambda|$ ways to choose $(S_1, \dots, S_\lambda) \in \cC_\lambda$. To construct $(T_1, \dots, T_\lambda)$, we first decide on one of $\binom{\lambda}{r}$ positions to disagree on.
        \item Within each disagreeing position $u \in [r]$, we choose where the $s_u$ agreements in $A_u$ occur, giving a factor of $\binom{p}{s_u}$.
        \item To construct the corresponding $B_u$ that disagrees with $A_u$ in the remaining $p-s_u$ positions, we need to choose the values in $B_u \setminus A_u$ out of the $n-\lambda p$ remaining values not in $S$. We need to choose a total of $p/2$ values across all $B_u \setminus A_u$, giving $\binom{n-\lambda p}{p/2}$. They are then assigned into blocks with occupancies $p-s_u$, giving a factor of $(p/2)! / \prod_{u=1}^r (p-s_u)!$.
    \end{itemize}
    This gives
    \begin{align}
        N_\disj(r,s) = |\cC_\lambda| \binom{\lambda}{r} \binom{n-\lambda p}{p/2} (p/2)! \prod_{u=1}^r \frac{\binom{p}{s_u}}{(p-s_u)!} = \binom{n}{p}^\lambda \lr{1 + O\lr{\frac{\lambda^2}{n}}} \binom{\lambda}{r} \binom{n-\lambda p}{p/2} (p/2)! \prod_{u=1}^r \frac{\binom{p}{s_u}}{(p-s_u)!}.
    \end{align}
    This gives for asymptotically large $\lambda = o(\sqrt n)$ the final quantity
    \begin{align}
        \E_\PP \bra{v}\cK\ket{v} = \beta C_{k,p} \binom{n}{p}^\lambda \mu^\lambda n^{p/2} \lambda^{p/2} (1+o(1)).
    \end{align}
\end{proof}

\begin{lemma}[Variance of certificate energy]\label{lem:kvar}
    For any $\ell = o(\sqrt n)$ and $\beta = \Omega(n^{-p/2} \log n)$, we have
    \begin{align}
        \Var \bra{v}\cK\ket{v} = o\lr{\lr{\E \bra{v}\cK\ket{v}}^2}.
    \end{align}
\end{lemma}
\begin{proof}
    We expand the second moment as
    \begin{align}
        \E_\PP \bra{v}\cK\ket{v}^2 &= \E_\PP \sum_{\substack{(S_1,\dots,S_\lambda)\in \cC_\lambda\\ (T_1,\dots,T_\lambda)\in \cC_\lambda}} \sum_{\substack{(S'_1,\dots,S'_\lambda)\in \cC_\lambda\\ (T'_1,\dots,T'_\lambda)\in \cC_\lambda}} \lr{\prod_{i=1}^\lambda f(x_{S_i})f(x_{T_i}) f(x_{S'_i})f(x_{T'_i})} \nonumber\\
        &\qquad \times \bra{S_1,\dots,S_\lambda}\cK\ket{T_1,\dots,T_\lambda} \bra{S'_1,\dots,S'_\lambda}\cK\ket{T'_1,\dots,T'_\lambda}\\
        &= \E_\PP \sum_{\substack{(S_1,\dots,S_\lambda)\in \cC_\lambda\\ (T_1,\dots,T_\lambda)\in \cC_\lambda\\ |S\triangle T|=p}} \sum_{\substack{(S'_1,\dots,S'_\lambda)\in \cC_\lambda\\ (T'_1,\dots,T'_\lambda)\in \cC_\lambda\\ |S'\triangle T'|=p}} \lr{\prod_{i \in \ag(S,T)}f(x_{S_i})^2}\lr{\prod_{i \in \ag(S',T')}f(x_{S'_i})^2}\nonumber\\
        &\qquad \times  \lr{\prod_{i\in\dis(S,T)} f(x_{S_i})f(x_{T_i})} \lr{\prod_{i\in\dis(S',T')} f(x_{S'_i})f(x_{T'_i})} \frac{(Y_{S\ominus T}-\theta_0) (Y_{S'\ominus T'}-\theta_0)}{\theta_0(1-\theta_0)}.
    \end{align}
    Let $r(S, T, S', T')$ denote the number of coordinates in $[n]$ that are in both $\cup_i S_i \cup_i T_i$ and $\cup_i S'_i \cup_i T'_i$.
    We deterministically bound
    \begin{align}
        \left|\lr{\prod_{i\in\dis(S,T)} f(x_{S_i})f(x_{T_i})} \lr{\prod_{i\in\dis(S',T')} f(x_{S'_i})f(x_{T'_i})}\right| \leq 1
    \end{align}
    and compute
    \begin{align}
        \E_{\PP(x)} \frac{(Y_{S\ominus T}-\theta_0) (Y_{S'\ominus T'}-\theta_0)}{\theta_0(1-\theta_0)} = \begin{cases}
            \beta^2 f(x_{S \ominus T}) f(x_{S' \ominus T'}) & S \ominus T \neq S' \ominus T'\\
            1 + \frac{1-2\theta_0}{\sqrt{\theta_0(1-\theta_0)}}\beta f(x_{S\ominus T}) & S\ominus T = S' \ominus T'.
        \end{cases}
    \end{align}
    We refer to the first case as ``off-diagonal" terms and the second case as ``diagonal" terms. We will upper-bound each case in a similar fashion, starting with the diagonal terms. For convenience, we will set $m = \lambda p + p/2$.
    
    Let $\cP$ denote the set of pairs $(S,T)$ with $S, T \in \cC_\lambda$ such that $|S \triangle T| = p$. Let $U = \cup_i S_i \cup_i T_i$ and $U' = \cup_i S'_i \cup_i T'_i$; then for $(S, T) \in \cP$ we have $|U| = m$. The condition $r(S,T,S',T')=r$ is equivalent to $|U\cap U'|=r$. Define for $C \in \binom{[n]}{p}$ counting factor
    \begin{align}
        A_U(C) = \#\{(S, T) \in \cP \,:\, S \ominus T=C, \, \cup_i S_i \cup_i T_i = U\}.
    \end{align}
    Then for $\cU_C = \{U\subset [n]\,:\, |U|=m, \, C \subset U\}$ we define counting factor
    \begin{align}
        N_\diag(r) &= \#\{(S,T,S',T') \in \cP^2 \,:\, S \ominus T = S' \ominus T', \, r(S,T,S',T')=r\} \\
        &= \sum_{C\in\binom{[n]}{p}} \sum_{U\in\cU_C} A_U(C) \sum_{U'\in\cU_C \,:\,|U\cap U'|=r} A_{U'}(C).
    \end{align}
    We count $|\cU_C| = \binom{n-p}{m-p}$ so
    \begin{align}
        A_U(C) = \frac{\sum_{U \in \cU_C} A_U(C)}{|\cU_C|} = \frac{A(C)}{\binom{n-p}{m-p}}
    \end{align}
    for
    \begin{align}
        A(C) = \#\{(S, T) \in \cP \,:\, S \ominus T=C\}.
    \end{align}
    We also count
    \begin{align}
        \#\{U' \in \cU_C \,:\, |U \cap U'|=r\} = \binom{m-p}{r - p} \binom{n-m}{m - r}
    \end{align}
    since outside the $p$ elements of $C$, we must choose $r-p$ elements from $U\setminus C$ to complete the overlap, and then the remaining elements of $U'$ from $[n]\setminus U$. Combining everything gives
    \begin{align}
        N_\diag(r) = \sum_{C\in\binom{[n]}{p}} \frac{A(C)^2}{\binom{n-p}{m-p}} \binom{m-p}{r-p}\binom{n-m}{m-r}.
    \end{align}
    We finish evaluating this with
    \begin{align}
        \sum_{C \in \binom{[n]}{p}} A(C)^2 = \frac{|\cP|^2}{\binom{n}{p}} = \Theta\lr{\binom{n}{p}^{2\lambda-1}n^p\lambda^p}
    \end{align}
    using the $|\cP| = \Theta(|\cC_\lambda|n^{p/2}\lambda^{p/2})$ from \Cref{lem:kmean} and $|\cC_\lambda| = \binom{n}{p}^\lambda$ from \Cref{lem:coverlap}. Hence, we find
    \begin{align}
        N_\diag(r) &= \begin{cases}
            0 & r < p\\
            O\lr{\binom{n}{p}^{2\lambda} n^{p-r} \lambda^p \binom{m-p}{r-p}} & r \geq p.
        \end{cases}
    \end{align}
    We now similarly bound the number of off-diagonal terms as
    \begin{align}
        N_\mathrm{offdiag}(r) &= \#\{(S,T,S',T') \in \cP^2 \,:\, S\ominus T \neq S' \ominus T',\, r(S,T,S',T')=r\} \\
        &\leq \#\{(S,T,S',T') \in \cP^2 \,:\, r(S,T,S',T')=r\}\\
        &\leq |\cP|^2 \binom{m}{r} \lr{\frac{m}{n}}^r\\
        &= O\lr{\binom{n}{p}^{2\lambda} n^p \lambda^p \binom{m}{r} \lr{\frac{m}{n}}^r}
    \end{align}
    by a union bound. In both the diagonal and off-diagonal cases we have by Cauchy-Schwarz, when $r(S,T,S',T')=r$, that
    \begin{align}
        \left|\E_x \lr{\prod_{i \in \ag(S,T)}f(x_{S_i})^2}\lr{\prod_{i \in \ag(S',T')}f(x_{S'_i})^2}\right| &\leq \mu^{|\ag(S,T)| + |\ag(S',T')| - 2r}\alpha^r\\
        &\leq \mu^{2\lambda - p}\lr{\frac{\alpha}{\mu^2}}^r
    \end{align}
    since $|\ag(S,T)| \in \{\lambda-p/2, \dots, \lambda-1\}$ and $\mu < 1$. Only $r \geq 1$ terms contribute to the variance, giving
    \begin{align}
        \Var \bra{v}\cK\ket{v} &\leq \mu^{2\lambda - p} \sum_{r=1}^m \lr{\frac{\alpha}{\mu^2}}^r \lr{\beta^2 N_\mathrm{offdiag}(r) + N_\diag(r)}\\
        &\leq \mu^{2\lambda - p} \binom{n}{p}^\lambda n^p \lambda^p \lr{\beta^2 \sum_{r=1}^m \binom{m}{r} \lr{\frac{m\alpha}{n\mu^2}}^r + \sum_{r=p}^m  \binom{m-p}{r-p} \lr{\frac{\alpha}{n \mu^2}}^r}\\
        &\leq \mu^{2\lambda - p} \binom{n}{p}^\lambda n^p \lambda^p \lr{\left[\beta^2 \lr{1 + \frac{m\alpha}{n\mu^2}}^m-1\right] + \lr{\frac{\alpha}{n \mu^2}}^p \lr{1+\frac{\alpha}{n \mu^2}}^{m-p}}\\
        &\leq \lr{\beta \binom{n}{p}^\lambda \mu^\lambda n^{p/2} \lambda^{p/2}}^2 \cdot O\lr{\frac{\lambda^2}{n} + \frac{1}{\beta^2n^p}}.
    \end{align}
    This gives $\Var \bra{v}\cK\ket{v} = o\lr{\lr{\E \bra{v}\cK\ket{v}}^2}$ for all $\lambda = o(\sqrt n)$ and $\beta \gg n^{-p/2}$; we set $\beta = \Omega(n^{-p/2}\log n)$ as a sufficient condition in the lemma statement.
\end{proof}

\begin{corollary}
\label{cor:planted_energy}
    For any $\ell = o(\sqrt n)$ and  $\beta = \Omega(n^{-p/2} \log n)$, we have with probability $1-o(1)$ that
    \begin{align}
        \norm{\cK} \geq C_{p,k} \beta n^{p/2} \ell^{p/2},
    \end{align} where $C_{p,k}$ is independent of $n$ and $\ell$.
    Note that the information-theoretic threshold is $\beta \gtrsim n^{(1-p)/2}$, which satisfies the above condition on $\beta$.
\end{corollary}
\begin{proof}
    Follows immediately from applying Chebyshev's inequality to the quantity $\bra{v}\cK\ket{v}$ (using the mean and variance of~\Cref{lem:kmean} and \Cref{lem:kvar}) and union bounding with the concentration bound of $\norm{v}^2$ shown in \Cref{lem:certnorm}.
\end{proof}

\subsection{Proof of the detection threshold}

The proof of \Cref{thm:kikuchi_bosonic} follows from combining \Cref{lem:bosonic_nullknorm} with \Cref{cor:planted_energy} proven below.
\begin{proof}[Proof of \Cref{thm:kikuchi_bosonic}] By \Cref{cor:planted_energy} and \Cref{lem:bosonic_nullknorm}, we have
  \begin{align}
        \Pr_\PP \left[\lambdamax{\cK} \leq C'_{p,k} \cdot \beta n^{p/2} \ell^{p/2}\right] = o(1), \quad \text{ and }  \quad \Pr_\QQ \left[  \norm{\cK} \geq  \sqrt{6 n^{p/2}  \ell^{1+p/2} \log(n)} \right] \leq o(1).
    \end{align}
Hence estimating the largest eigenvalue of the Kikuchi matrix up to an (arbitrary) multiplicative factor $0.49$ achieves detection provided \begin{equation}
    C'_{p,k} \cdot \beta n^{p/2} \ell^{p/2} \geq 2 \sqrt{6 n^{p/2}  \ell^{1+p/2} \log(n)}.
\end{equation} 
Rewriting this in terms of the signal to noise ratio $\beta = \frac{\eps}{\sqrt{\theta_0(1-\theta_0)}}$ shows that we achieve detection for any SNR \begin{align}
    \beta  &=\frac{2 \sqrt{6}}{C'_{p,k}} \sqrt{\frac{1}{n^{p/2}\ell^{p/2-1}} \log(n)},
\end{align} and setting $C' = \frac{2}{3}C_{k,p}$ completes the proof.
For $p > 2$, increasing $\ell$ thus allows detection at a smaller noise rate, and the detection threshold smoothly interpolates between the computational threshold $\beta \sim n^{-p/4}$ and the information-theoretic threshold $\beta \sim n^{(1-p)/2}.$
\end{proof}
In the next section, we show that our Kikuchi-based algorithm is optimal (up to a $\sqrt{\log(n)}$ factor, which can likely be removed by using stronger bounds than Matrix Bernstein).

\section{Lower bounds against low coordinate degree functions}\label{sec:low}
We now establish that our classical Kikuchi-based algorithm is asymptotically optimal (up to a $\sqrt{\log(n)}$ factor, which can likely be removed via a tighter concentration bound than Matrix Bernstein).
\subsection{Technical preliminaries}

We use the framework of~\cite{kunisky2024low} to show our LCDF lower bounds. These address a large class of models known as \emph{generalized stochastic block models} (GSBMs).

\begin{definition}[GSBM]\label{def:gsbm}
Let $p \geq 2$, let $k,n \geq1$, and let $\Omega$ be a measurable space. A generalized stochastic block model is specified by, for each $a \in [k]^p$, a probability measure $\mu_a$ on $\Omega$. Write
\begin{align}
    \mu_\mathrm{avg} = \frac{1}{k^p} \sum_{a \sim \unif([k]^p)} \mu_a
\end{align}
so that $\mu_\mathrm{avg}$ is another probability measure on $\Omega$. The GSBM then consists of the following two probability measures over $Y \in \Omega^{\binom{[n]}{p}}$:
\begin{enumerate}
    \item Under $\QQ$, draw $Y \sim \QQ$ with $Y_S \sim \mu_\mathrm{avg}$ independently for each $S \in \binom{[n]}{p}$.
    \item Under $\PP$, first draw $x \sim \unif([k]^n)$. Then for each $S = \{s_1 < \cdots < s_p\} \in \binom{[n]}{p}$, draw $Y_S \sim \mu_{x_{s_1},\dots,x_{s_p}}$ independently.
\end{enumerate}
\end{definition}

The lower bounds apply to GSBMs that satisfy the following assumptions.

\begin{assumption}\label{ass:gsbm}
A GSBM must satisfy the following.
\begin{enumerate}
    \item (Non-trivial.) There exists $a \in [k]^p$ such that $\mu_a \neq \mu_\mathrm{avg}$.
    \item (Regular.) For all $a \in [k]^p$, the likelihood ratio $d\mu_a/d\mu_\mathrm{avg}$ belongs to $L^2(\mu_\mathrm{avg})$.
    \item (Weakly symmetric.) For all $a, b \in [k]^p$ and all permutations of $\sigma \in \mathrm{Sym}([p])$,
    \begin{align}
        \E_{y\sim\mu_\mathrm{avg}}\left[ \frac{d\mu_{(a_1,\dots,a_p)}}{d\mu_\mathrm{avg}}(y) \cdot \frac{d\mu_{(b_1,\dots,b_p)}}{d\mu_\mathrm{avg}}(y)\right] = \E_{y\sim\mu_\mathrm{avg}}\left[ \frac{d\mu_{(a_{\sigma(1)},\dots,a_{\sigma(p)})}}{d\mu_\mathrm{avg}}(y) \cdot \frac{d\mu_{(b_{\sigma(1)},\dots,b_{\sigma(p)})}}{d\mu_\mathrm{avg}}(y)\right].
    \end{align}
    Note this is implied by a strongly symmetric GSBM, i.e., if for all $a \in [k]^p$ and $\sigma \in \mathrm{Sym}([p])$, $\mu_{a_1,\dots,a_p} = \mu_{a_{\sigma(1)}, \dots, a_{\sigma(p)}}$.
\end{enumerate}
\end{assumption}

The key object in showing the lower bounds is the characteristic tensor, which captures $\chi^2$ divergence on the diagonal and a cross-term $\chi^2$-like quantity between different $\mu_a$ relative to $\mu_\mathrm{av}$.

\begin{definition}[Characteristic tensor]
    For a GSBM specified by $(\mu_a)_{a\in[k]^p}$, we define its characteristic tensor to be $T = T^{(p)} \in (\R^{[k]\times[k]})^{\otimes p}$ having entries
    \begin{align}
        T_{(a_1,b_1),\dots,(a_p,b_p)} = \frac{1}{p!} \E_{y\sim\mu_\mathrm{avg}}\left[\lr{\frac{d\mu_{(a_1,\dots,a_p)}}{d\mu_\mathrm{avg}}(y)-1}\lr{\frac{d\mu_{(b_1,\dots,b_p)}}{d\mu_\mathrm{avg}}(y)-1}\right].
    \end{align}
\end{definition}

To identify information captured by marginals, we require the ability to contract the tensor.

\begin{definition}[Partial tensor contraction]
    Let $T \in (\R^N)^{\otimes p}$ be a symmetric tensor and let $v_1,\dots,v_m \in \R^N$ for some $1 \leq m \leq p$. We write $T[v_1,\dots,v_m,\cdot,\dots,\cdot] \in (\R^N)^{\otimes p-m}$ for the tensor with entries
    \begin{align}
        (T[v_1,\dots,v_m,\cdot,\dots,\cdot])_{i_1,\dots,i_{p-m}} = \sum_{j_1,\dots,j_m=1}^N T_{j_1,\dots,j_m,i_1,\dots,i_{p-m}}(v_1)_{j_1}\cdots (v_m)_{j_m}.
    \end{align}
\end{definition}

The SNR-computation tradeoff is then established by a quantity known as \emph{marginal order}.

\begin{definition}[Marginal order]\label{def:morder}
    For characteristic tensor $T^{(p)}$ of a GSBM as above, define the sequence of tensors $T^{(p-j)} \in (\R^{[k]\times[k]})^{\otimes(p-j)}$ by
    \begin{align}
        T^{(p-j)} = \frac{1}{k^{2j}} T^{(p)}[1,\dots,1,\cdot,\dots,\cdot]
    \end{align}
    for $j$ entries 1 and $p-j$ entries $\cdot$, where $1$ is the vector all of whose entries are 1 (in this case of dimension $k^2$).
    The marginal order of a GSBM is the smallest $p_*$ for which $T^{(p_*)} \neq 0$.
\end{definition}

Equivalently, the marginal characteristic tensor $T^{(p-j)}$ is the characteristic tensor of another GSBM formed by marginalizing $\mu_a$ of the original GSBM. That is, it is the characteristic tensor of a model with
\begin{align}
    \mu_{a_1,\dots,a_{p-j}}^{(p-j)} = \frac{1}{k^j} \sum_{a_{p-j+1},\dots,a_p=1}^k \mu_{a_1,\dots,a_p}.
\end{align}
Operationally, one can sample from this distribution by extending it to length $p$ with uniformly random entries and sampling from $\mu_a$. The marginal order obstructs LCDF, as given by the following recent result.

\begin{theorem}[Theorem 1.13 of~\cite{kunisky2024low}]\label{thm:kun}
    Consider a GSBM with marginal order $p_* \geq 2$, and denote its marginal characteristic tensor by $T^{(j)}$. There is a constant $c = c_{k,p}$ depending only on $p$ and $k$ such that, if for all sufficiently large $n$ we have $D(n) \leq cn$ and
    \begin{align}
        \max_{p_* \leq j \leq p} \max_{\norm{v}=1} \left|\sum_{i_1,\dots,i_j} T^{(j)}_{i_1,\dots,i_j} v_{i_1}\cdots v_{i_j}\right| \leq cn^{-p+p_*/2} D(n)^{1-p_*/2}
    \end{align}
    then no sequence of functions of coordinate degree at most $D(n)$ can strongly separate $\mathbb Q_n$ from $\mathbb P_n$.
\end{theorem}

\subsection{Lower bounds for \texorpdfstring{$p$}\ -marginal Bernoulli GSBMs}

We prove \Cref{thm:ldm_lowerbounds} by showing that the $p$-marginal HSBM is a GSBM with marginal order $p$ satisfying \Cref{ass:gsbm}, and then obtaining a lower bound via \Cref{thm:kun}. In fact, our lower bounds will hold for a larger class of models that we call $p$-marginal Bernoulli GSBMs; we will then show that a $p$-marginal HSBM is a type of $p$-marginal Bernoulli GSBM. These models look like a $p$-marginal HSBM but with a generic choice of $f$ that satisfies some conditions.

\begin{definition}[$p_*$-marginal Bernoulli GSBM]\label{def:pberngsbm}
Let $p \geq 2$, let $k,n \geq1$, and let $\theta_0 \in (0, 1)$. Let $\epsilon \in \R$ and $f:[k]^p \to \R$ be a symmetric function satisfying
\begin{align}
    \E_{a \sim \unif([k]^p)}[f(a)] = 0, \qquad 0 < \theta_0 + \epsilon \norm{f}_\infty < 1 , \qquad \exists\, a \in [k]^p \text{ s.t. } f(a) \neq 0
\end{align}
and for all $a_1,\dots,a_r \in [k]$ and all $r < p_*$
\begin{align}
    \E_{a_{r+1},\dots,a_p \in [k]}[f(a_1,\dots,a_p) \,|\, a_1,\dots,a_r] = 0
\end{align}
where $p_* \geq 2$. Define (for $a \in [k]^p$) the probability measures
\begin{align}
    \mu_a = \bern(\theta_0 + \epsilon f(a)), \qquad \mu_\mathrm{avg} = \bern(\theta_0).
\end{align}
The $p_*$\emph{-marginal Bernoulli GSBM} then consists of the following two probability measures over $Y \in \{0,1\}^{\binom{[n]}{p}}$:
\begin{enumerate}
    \item Under $\QQ$, draw $Y \sim \QQ$ with $Y_S \sim \mu_\mathrm{avg}$ independently for each $S \in \binom{[n]}{p}$.
    \item Under $\PP$, first draw $x \sim \unif([k]^n)$. Then for each $S = \{s_1 < \cdots < s_p\} \in \binom{[n]}{p}$, draw $Y_S \sim \mu_{x_{s_1},\dots,x_{s_p}}$ independently.
\end{enumerate}
\end{definition}

We show that this model satisfies \Cref{ass:gsbm} and has a marginal order $p_*$.

\begin{lemma}[Marginal order of $p_*$-marginal Bernoulli GSBM]\label{lem:pberngsbm}
    A $p_*$-marginal Bernoulli GSBM is a non-trivial, regular, strongly symmetric GSBM with marginal order $p_*$.
\end{lemma}
\begin{proof}
    Since $\EE{f(a)} = 0$, the probability measures above satisfy $\mu_\mathrm{avg} = \frac{1}{k^p} \sum_{a \sim \unif([k]^p)} \mu_a$ and thus a Bernoulli GSBM is indeed a special case of a GSBM (Definition~\ref{def:gsbm}). It is immediately symmetric (since $f$ is symmetric) and non-trivial (since $f(a) \neq 0$ for some a) and regular (since $0 < \theta_0 + \epsilon \norm{f}_\infty < 1$). We now compute its marginal order.
    
    The marginal characteristic tensor $T^{(r)}$ is the characteristic tensor of
    \begin{align}
        \mu_{a_1,\dots,a_r}^{(r)} &= \frac{1}{k^{p-r}} \sum_{a_{r+1},\dots,a_p\in [k]} \mu_{a_1,\dots,a_p}.
    \end{align}
    This is another Bernoulli distribution with variable $\theta_{a_1,\dots,a_r}$ given by
    \begin{align}
        \theta_{a_1,\dots,a_r} = \frac{1}{k^{p-r}} \sum_{a_{r+1},\dots,a_p\in [k]} \theta_{a_1,\dots,a_p}.
    \end{align}
    For $\theta_a = \theta_0 + f(a_1,\dots,a_p)$, we have
    \begin{align}
        \mu_{a_1,\dots,a_r}^{(r)} &= \theta_0 + \frac{1}{k^{p-r}} \sum_{a_{r+1},\dots,a_p\in [k]} f(a) = \theta_0 + f^{(r)}(a_1,\dots,a_r).
    \end{align}
    Since $\EE{f(a)} = 0$ by assumption, we have that $\EE{\mu_a} = \mu_\mathrm{avg}$ and that
    \begin{align}
        \frac{d \mu_{a_1,\dots,a_r}^{(r)}}{d \mu_\mathrm{avg}}(1) - 1 = \frac{f^{(r)}(a_1,\dots,a_r)}{\theta_0}, \quad \frac{d \mu_{a_1,\dots,a_r}^{(r)}}{d \mu_\mathrm{avg}}(0) - 1 = -\frac{f^{(r)}(a_1,\dots,a_r)}{1-\theta_0}.
    \end{align}
    Hence, the marginal characteristic tensor has elements
    \begin{align}
        \EE{\lr{\frac{d \mu_{a_1,\dots,a_r}^{(r)}}{d \mu_\mathrm{avg}}(y) - 1}\lr{\frac{d \mu_{b_1,\dots,b_r}^{(r)}}{d \mu_\mathrm{avg}}(y) - 1}} = \lr{\frac{1}{\theta_0} + \frac{1}{1-\theta_0}} f^{(r)}(a_1,\dots,a_r)f^{(r)}(b_1,\dots,b_r).
    \end{align}
    A sufficient condition for the marginal order to be $\geq p_*$ is
    \begin{align}
        f^{(r)}(a_1,\dots,a_r) = \frac{1}{k^{p-r}} \sum_{a_{r+1},\dots,a_p} f(a_1,\dots,a_p) = 0
    \end{align}
    for all $a_1, \dots, a_r \in [k]$ and $r < p_*$.
\end{proof}

To show that a $p$-marginal HSBM is a $p$-marginal Bernoulli GSBM, we introduce a generic procedure for constructing $f$ that satisfies the conditions of $f$ in the definition of a $p$-marginal Bernoulli GSBM. This procedure decomposes $f$ over its marginals and sets marginals smaller than $p_*$ to zero.

\begin{lemma}[$p_*$-whitening]\label{lem:white}
    Let $p, k \geq 2$ and let $\E$ denote expectations over uniform $a \in [k]^p$. Any $\tilde f:[k]^p \to \R$ can be decomposed as
    \begin{align}
        \tilde f(a_1,\dots,a_p) = \sum_{S\subseteq [p]} \sum_{T\subseteq S} (-1)^{|S\setminus T|} \EE{\tilde f(a)|a_T}.
    \end{align}
    For $p \geq p_* \geq 2$, the $p_*$-\emph{whitened} function
    \begin{align}
        f(a_1,\dots,a_p) = \sum_{\substack{S\subseteq [p]\\ |S| \geq p_*}} \sum_{T\subseteq S} (-1)^{|S\setminus T|} \EE{\tilde f(a)|a_T}
    \end{align}
    satisfies, for all $R \subseteq [p]$ such that $|R| < p_*$,
    \begin{align}
        \EE{f(a_1,\dots,a_p) \,|\, a_R} = 0.
    \end{align}
\end{lemma}
\begin{proof}
    For each $S \subseteq [p]$, define
    \begin{align}\label{eq:fsdef}
        m_S(a_S) = \EE{\tilde f(a)\,|\,a_S}, \quad \tilde f_S(a_S) = \sum_{T\subseteq S} (-1)^{|S\setminus T|} m_T(a_T).
    \end{align}
    We will use the identity, for any set $C$, that
    \begin{align}\label{eq:setid}
        \sum_{X \subseteq C} (-1)^{|X|} = 1_{C = \emptyset}.
    \end{align}
    This implies that, for any $R \subseteq [p]$,
    \begin{align}\label{eq:mr}
        \sum_{S \subseteq R} \tilde f_S(a_S) &= \sum_{T\subseteq R} m_T(a_T) \sum_{S\supseteq T} (-1)^{|S\setminus T|} = \sum_{T\subseteq R} m_T(a_T) \sum_{U \subseteq R\setminus T} (-1)^{|U|} = \sum_{T\subseteq R} m_T(a_T) 1_{T = R} \nonumber\\
        &= m_R(a_R) = \EE{\tilde f(a)\,|\,a_R}.
    \end{align}
    Choosing $R=[p]$ so $\EE{\tilde f(a)\,|\,a} = \tilde f(a)$ proves our first claim that
    \begin{align}
        \tilde f(a_1,\dots,a_p) = \sum_{S\subseteq [p]} \tilde f_S(a_S).
    \end{align}
    We then whiten $\tilde f$ by setting $m_T$ with $|T| < p_*$ to zero. We claim that
    \begin{align}
        f(a_1,\dots,a_p) = \sum_{\substack{S\subseteq [p]\\ |S| \geq p_*}} \tilde f_S(a_S)
    \end{align}
    satisfies for all $|R| < p_*$
    \begin{align}\label{eq:fcond}
        \EE{f(a_1,\dots,a_p) \,|\, a_R} = 0.
    \end{align}
    To obtain this result, it suffices to show that
    \begin{align}\label{eq:fs0}
        \EE{\tilde f_S(a_S) \,|\, a_R} = 0.
    \end{align}
    By the definition of $\tilde f_S$ \eqref{eq:fsdef}, we have that
    \begin{align}
        \EE{\tilde f_S(a_S) \,|\, a_R} &= \sum_{T\subseteq S} (-1)^{|S\setminus T|} \EE{m_T(a_T) \,|\, a_R}.
    \end{align}
    Since $m_T$ only depends on the coordinates in $T$, we can rewrite this as
    \begin{align}
        \EE{\tilde f_S(a_S) \,|\, a_R} &= \sum_{T\subseteq S} (-1)^{|S\setminus T|} \EE{m_T(a_T) \,|\, a_{R\cap T}}.
    \end{align}
    We separate out $U = R \cap T$ to obtain
    \begin{align}
        \EE{\tilde f_S(a_S) \,|\, a_R} &= \sum_{U \subseteq R} \sum_{\substack{T\subseteq S\\ R \cap T = U}} (-1)^{|S\setminus T|} \EE{m_T(a_T) \,|\, a_U}.
    \end{align}
    Due to the expectation, $\EE{m_T(a_T) \,|\, a_U} = g_U(a_U)$ is only a function of $U$ and $a_U$, so
    \begin{align}
        \EE{\tilde f_S(a_S) \,|\, a_R} &= \sum_{U \subseteq R} \sum_{\substack{T\subseteq S\\ R \cap T = U}} (-1)^{|S\setminus T|} g_U(a_U)\\
        &= \sum_{U \subseteq R} g_U(a_U) \sum_{T' \subseteq S \setminus R} (-1)^{|S\setminus (U \cup T')|}\\
        &= \sum_{U \subseteq R} g_U(a_U) (-1)^{|S\setminus U|} \sum_{T' \subseteq S \setminus R} (-1)^{|T'|}.
    \end{align}
    By \eqref{eq:setid}, we have
    \begin{align}
        \sum_{T' \subseteq S \setminus R} (-1)^{|T'|} = 0
    \end{align}
    since $S \setminus R$ is never empty due to the condition $|S| \geq p_*$ and $|R| < p_*$. Hence, \eqref{eq:fs0} is satisfied, implying our claim \eqref{eq:fcond}.
\end{proof}

\begin{remark}
    Up to checking \Cref{ass:gsbm}, a $p_*$-whitened $f$ satisfies the conditions needed for a $p_*$-marginal Bernoulli GSBM.
\end{remark}

By applying this whitening procedure to obtain a marginal order of $p_* > 2$, we obtain a generic method to produce an SBM-like model with an SNR-computation tradeoff by \Cref{thm:kun}. We anticipate that these models also generally have matching upper bounds via the Kikuchi method, although our proof strategy is specialized to the $p$-marginal HSBM. We also expect whitening to be a generic approach to construct toy models for datasets that cannot be learned from two-body marginals.

\subsection{Lower bounds and properties of \texorpdfstring{$p$}\ -marginal HSBM}
The $p$-marginal HSBM can be constructed by $p$-whitening a choice of $f$ that sets the probability of a hyperedge on a set of $p$ vertices based on if the set is all part of the same community.

\begin{lemma}[Whitened indicator function]\label{lem:fprop}
    Let $k,p \geq 2$. Let $\tilde f(a) = 1_{a_1=\cdots=a_p} - k^{1-p}$. Then $p$-whitening (Lemma~\ref{lem:white}) produces the function
    \begin{align}\label{eq:white}
        f(a_1, \dots, a_p) = \sum_{i=1}^k \prod_{j=1}^p \lr{1_{a_j=i} - \frac{1}{k}},
    \end{align}
    which satisfies the following properties.
    \begin{enumerate}
        \item $\sum_{a_1,\dots,a_p} f(a_1,\dots,a_p) = 0$.
        \item $\sum_{a_{r+1},\dots,a_p} f(a_1,\dots,a_p) = 0$ for all $r < p$.
        \item $f$ is symmetric with respect to its arguments.
        \item $\exists \, a \in [k]^p$ such that $f(a) \neq 0$.
        \item If there are $d \geq 2$ distinct values in $(a_1,\dots,a_p)$ then $|f(a)| \leq 2/k$. Moreover, $1 > f(a_1,\dots,a_p) > 1 - \frac{p+1}{k}$ for $a_1=\cdots=a_p$.
    \end{enumerate}
\end{lemma}
\begin{proof}
    We apply the whitening procedure of Lemma~\ref{lem:white} with $S = [p]$ to set $p_* = p$:
    \begin{align}
        f(a_1,\dots,a_p) = \sum_{T\subseteq [p]} (-1)^{|[p]\setminus T|} \EE{1_{a_1=\cdots = a_p} - \frac{1}{k^{p-1}} \,|\, a_T}.
    \end{align}
    We evaluate the expectation as
    \begin{align}
        \EE{1_{a_1=\cdots = a_p} \,|\, a_T} = \begin{cases}
            \frac{1}{k^{p-|T|}} & a_t \text{ equal for all } t \in T\\
            0 & \text{else}
        \end{cases} = \sum_{i=1}^k \frac{1}{k^{p-|T|}} \prod_{j \in T} 1_{a_j = i},
    \end{align}
    giving (using the identity \eqref{eq:setid} and the binomial theorem)
    \begin{align}
        f(a_1,\dots,a_p) &= \sum_{T\subseteq [p]} (-1)^{p-|T|} \lr{ - \frac{1}{k^{p-1}} + \sum_{i=1}^k \frac{1}{k^{p-|T|}} \prod_{j \in T} 1_{a_j = i}} \\
        &= \sum_{T\subseteq [p]} (-1)^{p-|T|} \sum_{i=1}^k \frac{1}{k^{p-|T|}} \prod_{j \in T} 1_{a_j = i}\\
        &= \sum_{i=1}^k \sum_{T \subseteq [p]} \lr{-\frac{1}{k}}^{p-|T|} \prod_{j \in T} 1_{a_j=i} \\
        &= \sum_{i=1}^k \prod_{j=1}^p \lr{1_{a_j=i} - \frac{1}{k}}
    \end{align}
    as claimed in the lemma statement.
    It is also useful to rewrite $f$ as
    \begin{align}\label{eq:frewrite}
        f(a) = \sum_{i=1}^k \lr{1 - \frac{1}{k}}^{m_i} \lr{-\frac{1}{k}}^{p-m_i}, \qquad m_i = |\{j \leq p : a_j = i\}|.
    \end{align}
    The first three properties are immediate with the aid of Lemma~\ref{lem:white}.
    To show that a nontrivial choice of $a$ exists, we check that
    \begin{align}\label{eq:f11}
        f(1,\dots,1) = \lr{1-\frac{1}{k}}^p + (k-1)\lr{-\frac{1}{k}}^p \neq 0.
    \end{align}
    We now show both parts of the last property. To analyze $a_1=\cdots=a_p$, we set without loss of generality all $a_i=1$. We apply Bernoulli's inequality $(1-x)^n \geq 1 - nx$ to \eqref{eq:f11} to obtain that $1 > f(1,\dots,1) \geq 1 - \frac{p}{k} + (k-1) (-k)^{-p}$. Since $|(k-1) (-k)^{-p}| < k^{-p+1} \leq k^{-1}$, we have that $f(1,\dots,1) > 1 - \frac{p+1}{k}$. To analyze $d$ distinct values, it suffices to bound $|f(a)|$ by the largest case: $d=2$ and $p-1$ values are equal. Since $|f(a)| \leq k^{-p}\sum_{i=1}^k (k-1)^{m_i}$, this choice gives $|f(a)| \leq (k-1)^{-1} + (k-1)k^{-p} \leq 2/k$.
\end{proof}

\begin{corollary}[Properties of $p$-marginal HSBM]
    A $p$-marginal HSBM is a non-trivial, regular, strongly symmetric GSBM with marginal order $p$ by Lemmas~\ref{lem:pberngsbm} and~\ref{lem:fprop}. Moreover, its $f$ satisfies
    \begin{align}
        f(a) \to 1_{a_1 = \cdots = a_p} \quad \mathrm{as} \quad k\to\infty,
    \end{align}
    giving a community detection interpretation where vertices in the same community have a larger probability of receiving a hyperedge.
\end{corollary}

We conclude by proving \Cref{thm:ldm_lowerbounds} for the $p$-marginal HSBM by bounding the injective norm of its marginal characteristic tensor.

\begin{lemma}[LCDF lower bound for $p$-marginal HSBM]
    Let $\PP_n$ and $\QQ_n$ be the planted and null distributions of a $p$-marginal HSBM on $n$ vertices with $p \geq 3$. No sequence of functions of coordinate degree at most $D(n)$ can strongly separate $\QQ_n$ from $\PP_n$ for
    \begin{align}
        D(n) \lesssim \lr{\frac{\epsilon}{\sqrt \theta_0}}^{-1/(p-2)} n^{-1/4 + 1/2p}.
    \end{align}
\end{lemma}
\begin{proof}
    Writing
    \begin{align}
        f(a) = \sum_{i=1}^k \lr{1 - \frac{1}{k}}^{m_i} \lr{-\frac{1}{k}}^{p-m_i}, \qquad m_i = |\{j \leq p : a_j = i\}|,
    \end{align}
    the characteristic tensor has elements
    \begin{align}
        T_{a,b} = \frac{\epsilon^2}{p!} \lr{\frac{1}{\theta_0} + \frac{1}{1-\theta_0}} f(a) f(b) = \frac{\epsilon^2}{p!} \lr{\frac{1}{\theta_0} + \frac{1}{1-\theta_0}} \sum_{i,j=1}^k \lr{1 - \frac{1}{k}}^{m_i(a) + m_j(b)} \lr{-\frac{1}{k}}^{2p-m_i(a) - m_j(b)}
    \end{align}
    and thus, choosing $v_{(1,1)} = 1$ and 0 elsewhere,
    \begin{align}
        \max_{\norm{v}=1} \left|\sum_{a,b \in [k]^p} T_{a,b} v_{(a_1,b_1)}\cdots v_{(a_p,b_p)}\right| &\geq \frac{\epsilon^2}{p!} \lr{\frac{1}{\theta_0} + \frac{1}{1-\theta_0}}\left[\lr{1 - \frac{1}{k}}^p + (k-1)k^{-p}\right]^2.
    \end{align}
    Asymptotically, for $\theta_0 < 1/2$, the relation
    \begin{align}
        \max_{\norm{v}=1} \left|\sum_{a,b \in [k]^p} T_{a,b} v_{(a_1,b_1)}\cdots v_{(a_p,b_p)}\right| \leq cn^{-p/2} D(n)^{1-p/2}.
    \end{align}
    from \Cref{thm:kun} then gives threshold
    \begin{align}
        \beta \leq C'(k,p) n^{-p/4} D(n)^{1/2-p/4}
    \end{align}
    for some constant $C'(k,p)$ independent of $n$ and $\beta$.
\end{proof}

\section{Correctness and efficiency of the quantum Kikuchi algorithm}
\label{sec:qalg}
In this section, we describe a quantum algorithm that achieves a quartic quantum speedup over the classical Kikuchi method presented in \Cref{sec:Classical}. The core idea is to instantiate hypergraph community detection as an instance of a fine-grained version of the guided sparse Hamiltonian problem. 
\subsection{Quantized Kikuchi algorithm}
Recall our quantum algorithm from \Cref{sec:Quantum}.

\begin{center}
    \textbf{ Quantized Kikuchi method}
\end{center}
\noindent\hrulefill\\
\noindent\textbf{Input:} A $p$-marginal HSBM$(n,k,\theta_0,\eps)$ instance $\mathbf{Y}$.  \\
\noindent\textbf{Preprocessing:} As in the classical Kikuchi algorithm, choose a sufficiently large $\ell$ and a threshold $\tau$ in accordance with a classical ``Kikuchi Theorem'' at SNR $\beta$, such as \Cref{thm:kikuchi_bosonic}.\\
\textbf{Quantum algorithm:}
Encode the following in Amplitude Amplification and repeat $O(n^{\ell/4})$ times:
\begin{itemize}[left = 0.2cm]
        \item Prepare $\ell/p$ unentangled copies of a guiding state $\ket{\phi}$. Symmetrize the resulting state to obtain a guiding state $\ket{\Phi}$. 
        \item Perform Quantum Phase Estimation with the sparse Hamiltonian $\cK_\ell$ on the initial state $\ket{\Phi}$.
        \item Measure the eigenvalue register and record whether an eigenvalue above the threshold $\tau$ was sampled.
\end{itemize}
\textbf{Output:} If during any of the repetitions, an eigenvalue above $\tau$ was found, return ``Planted''. Otherwise, return ``Random''. 

\noindent\hrulefill

A main technical analysis of this algorithm is already provided by the classical discussion in \Cref{sec:classical_support}.
The main challenge we address in this section is the existence of an efficiently preparable guiding state that has improve overlap with the leading eigenspace of the Kikuchi matrix, that is, the eigenspace that certifies the presence of a community structure. We show this in two main steps: First, that there exists an efficiently preparable guiding state $\ket{u}$ that has improved overlap with the certificate vector $\ket{v}$.  \begin{definition}[Guiding vector] \label{def:guide_vector}
    Let $\ket{u'}$ be the unnormalized vector given by
    \begin{align}
        \ket{u'} = \sum_{S \in \binom{[n]}{p}} \frac{Y_S - \theta_0}{\sqrt{\theta_0(1-\theta_0)}} \ket{S}.
    \end{align}
    The guiding vector $\ket{u}$ is the unnormalized vector $\Pi_\lambda \ket{u'}^{\otimes \ell/p}$.
\end{definition}
Then, we prove lower bounds on the overlap between a slight modification of this guiding state and the leading eigenspace of the Kikuchi matrix. Our final result is summarized below. 
\begin{theorem}\label{thm:final_kikuchi_bosonic}
Consider the $p$-marginal HSBM with $p$ even. Let $\ell \in[p / 2, n-p / 2]$ and $\epsilon, \theta_0$ in accordance with \eqref{eq:paramconds_sec}, and let 
\begin{equation}
    \beta := \frac{\epsilon}{\sqrt{\theta_0(1-\theta_0)}} \geq \frac{3\sqrt{6}}{C_{k,p}} \ell^{1/2 - p/4} n^{-p/4}\sqrt{\log(n)},
\end{equation} in accordance with the classical Kikuchi theorem \Cref{thm:kikuchi_bosonic}. Let $\cKt$ be the $\ell$-th order bosonic Kikuchi matrix obtained by sample splitting as described in \Cref{sec:sample_splitting}. Let $\Pi_\geq(\cKt)$ denote the projector onto the eigenvectors of $\cKt$ with eigenvalues at least $\tau.$ Then there exists an efficiently preparable guiding state $\ket{\Tilde u}$ such that 
\begin{equation}
    \Pr_\PP\left[    \frac{\bra{\Tilde u} \Pi_\geq(\cKt) \ket{\Tilde u}}{\braket{\Tilde u}{\Tilde u}} \geq \exp(-\Tilde{O}(\ell))
    \cdot \Tilde{O}\left( n^{-\ell/2} \right) \cdot \Tilde{O}\left( \log(n)^{-\ell/p} \right) \right] \geq 1-o(1).
\end{equation}
\end{theorem}
The exponential factors in $\ell$ and polynomial factors in $n$ are negligible compared to the overall scaling of $n^{O(\ell)}$. The inverse of the square root of the guiding state overlap is the source of our quartic quantum speedup. We prove in \Cref{sec:symmetrization} that the guiding state can be prepared efficiently. 
\subsection{Guiding state overlap}
The following moments will be useful.
\begin{lemma}\label{lem:momentquantities}
    For sets $S, T \in \binom{[n]}{p}$, define
    \begin{align}
        \mu = \EE{f(x_S)^2} , \qquad \gamma = \EE{f(x_S)^3}, \qquad \alpha = \EE{f(x_S)^4}, \qquad \beta_r = \EE{f(x_S)^2 f(x_T)^2} \text{ for } |S\cap T| = r
    \end{align}
    where expectations are over uniform $x \in [k]^n$.
    Then for even $p$,
    \begin{align}
        \mu &= \frac{(k-1)^p+(-1)^p(k-1)}{k^{2p-1}}\\
        \gamma &= k^{-2p}\lr{(k-1)^p(k-2)^p + 3(k-1)(k-2)^p + (k-1)(k-2)2^p}\\
        \alpha &= k \Bigg((k-3) (k-2) (k-1) 3^p k^{-4 p}+4 (k-1) \left(\frac{(k-1)^3+1}{k^5}\right)^p+\left(\frac{(k-1)^4+k-1}{k^5}\right)^p\nonumber\\
        &\qquad\quad +6 (k-2) (k-1) \left(\frac{k-3}{k^4}\right)^p+3 (k-1) \left(\frac{2 k-3}{k^4}\right)^p\Bigg)\\
        \beta_r &\leq \alpha.
    \end{align}
\end{lemma}
\begin{proof}
    The computations of $\mu, \gamma, \alpha$ are direct evaluations. Cauchy-Schwarz gives
    \begin{align}
        \beta_r \leq \sqrt{\EE{f(x_S)^4} \EE{f(x_T)^4}} = \alpha.
    \end{align}
\end{proof}

\begin{lemma}[Certificate state norm]\label{lem:certnorm}
    For $\ell = o(\sqrt n)$ and any constant $\delta > 0$
    \begin{align}
        \pr{\left|\norm{v}^2 - \binom{n}{p}^\lambda \mu^\lambda\right| \geq \delta \binom{n}{p}^\lambda \mu^\lambda} = o(1).
    \end{align}
\end{lemma}
\begin{proof}
    The first moment is
    \begin{align}
        \E \norm{v}^2 = \E_x \sum_{(S_1, \dots, S_\lambda) \in \cC_\lambda} \prod_{i=1}^\lambda f(x_{S_i})^2 = |\cC_\lambda| \mu^\lambda
    \end{align}
    by independence of the $x$ coordinates. The second moment is computed via \Cref{lem:coverlap} using
    \begin{align}
        \E \norm{v}^4 = (\E \norm{v}^2)^2 + \E \sum_{r=1}^{\lambda p} \sum_{C,C' \in \cC_\lambda \,:\, r(C, C') = r} \prod_{i=1}^\lambda f(x_{C_i})^2 f(x_{C'_i})^2.
    \end{align}
    Since at most $r$ blocks agree when $r(C, C')=r$, we have that
    \begin{align}
        \E \prod_{i=1}^\lambda f(x_{C_i})^2 f(x_{C'_i})^2 \leq \lr{\frac{\alpha}{\mu^2}}^r \mu^{2\lambda}
    \end{align}
    for $\alpha = \E f(x_S)^4$ computed in \Cref{lem:momentquantities}. Hence,
    \begin{align}
        \E \norm{v}^4 &\leq (\E \norm{v}^2)^2 + |\cC_\lambda|^2 \mu^{2\lambda} \sum_{r=1}^{\lambda p} \binom{\lambda p}{r}\lr{\frac{\alpha \lambda p}{\mu^2 n}}^r\\
        &= (\E \norm{v}^2)^2 + |\cC_\lambda|^2 \mu^{2\lambda} \left[\lr{1 + \frac{\alpha \lambda p}{\mu^2 n}}^{\lambda p} - 1\right].
    \end{align}
    For $\ell = o(\sqrt n)$, we thus obtain variance
    \begin{align}
        \Var \norm{v}^2 &\leq (\E \norm{v}^2)^2 \left[\lr{1 + \frac{\alpha \lambda p}{\mu^2 n}}^{\lambda p} - 1\right] = (\E \norm{v}^2)^2 \cdot O\lr{\frac{\ell^2}{n}} = o\lr{(\E \norm{v}^2)^2}.
    \end{align}
    Chebyshev's inequality thus gives the claimed statement.
\end{proof}

\begin{lemma}[Guiding state norm]\label{lem:guidenorm}
    For $\ell = o(\sqrt n)$ and any constant $\delta > 0$,
    \begin{align}
        \pr{\left|\norm{u}^2 - \binom{n}{p}^\lambda\right| \geq \delta \binom{n}{p}^\lambda} = o(1).
    \end{align}
\end{lemma}
\begin{proof}
    Direct computation for random variable $u'_S = (Y_S - \theta_0)/\sqrt{\theta_0(1-\theta_0)}$ gives $\E_\PP (u'_S)^2 = 1$ and
    \begin{align}
        \alpha' := \E_\PP (u'_S)^4 = \frac{1}{\theta_0} + \frac{1}{1-\theta_0} - 3
    \end{align}
    and thus by Cauchy-Schwarz, for any $S, T \in \binom{[n]}{p}$,
    \begin{align}
        \E_\PP (u'_S)^2 (u'_T)^2 \leq \sqrt{\E_\PP (u'_S)^4 \E_\PP (u'_T)^4} \leq \alpha'.
    \end{align}
    The first moment is, by independence of edges and coordinates of $x$,
    \begin{align}
        \E_\PP \norm{u}^2 = \sum_{(S_1,\dots,S_\lambda)\in\cC_\lambda} \prod_{i=1}^\lambda \E_\PP (u'_{S_i})^2 = |\cC_\lambda|.
    \end{align}
    Note that by \Cref{lem:coverlap}, for $\ell = o(\sqrt n)$ we have
    \begin{align}
        |\cC_\lambda| = \binom{n}{p}^\lambda \lr{1 + o(1)}.
    \end{align}
    The second moment is
    \begin{align}
        \E_\PP \norm{u}^4 &= \lr{\E_\PP \norm{u}^2}^2 + \E_\PP \sum_{r=1}^{\lambda p} \sum_{C,C' \in \cC_\lambda \,:\, r(C, C') = r} \prod_{i=1}^\lambda (u'_{C_i})^2 (u'_{C'_i})^2\\
        &\leq (\E_\PP \norm{u}^2)^2 + |\cC_\lambda|^2  \left[\lr{1 + \frac{\alpha' \lambda p}{n}}^{\lambda p} - 1\right]
    \end{align}
    by \Cref{lem:coverlap} following similar steps to the proof of \Cref{lem:certnorm}. For $\ell = o(\sqrt n)$, we thus obtain variance
    \begin{align}
        \Var \norm{u}^2 &\leq (\E \norm{u}^2)^2 \left[\lr{1 + \frac{\alpha' \lambda p}{n}}^{\lambda p} - 1\right] = (\E \norm{u}^2)^2 \cdot O\lr{\frac{\ell^2}{n}} = o\lr{(\E \norm{u}^2)^2}.
    \end{align}
    Chebyshev's inequality thus gives the claimed statement.
\end{proof}

\begin{lemma}[Guiding state overlap]\label{lem:guideoverlap}
    For $\ell = o(\sqrt n)$ and any constant $\delta > 0$, we have
    \begin{align}
        \pr{\left|\frac{\bra{u}\ket{v}}{\norm{u}\norm{v}} - \lr{\beta\sqrt\mu}^\lambda\right| \geq \delta \lr{\beta\sqrt\mu}^\lambda} = o(1).
    \end{align}
    In particular, choosing $\beta = \tilde \Theta(\ell^{1/2-p/4}n^{-p/4})$ gives
    \begin{align}
        \pr{\frac{\bra{u}\ket{v}}{\norm{u}\norm{v}} \leq \tilde O\lr{n^{-\ell/4} \ell^{\ell(1/2p-1/4)}}} = o(1).
    \end{align}
\end{lemma}
\begin{proof}
    We have first moment
    \begin{align}\label{eq:uv1}
        \E_\PP \bra{u}\ket{v} = \E_\PP \sum_{(S_1, \dots, S_\lambda) \in \cC_\lambda} \prod_{i=1}^\lambda u'_{S_i} f(x_{S_i}) = \beta^\lambda \sum_{(S_1, \dots, S_\lambda) \in \cC_\lambda} \prod_{i=1}^\lambda \E_x f(x_{S_i})^2 = |\cC_\lambda| (\beta \mu)^\lambda.
    \end{align}
    To show the second moment, we compute
    \begin{align}
        \E_\PP \bra{u}\ket{v}^2 &= \E_\PP \sum_{\substack{(S_1, \dots, S_\lambda) \in \cC_\lambda\\ (T_1, \dots, T_\lambda) \in \cC_\lambda}} \lr{\prod_{i=1}^\lambda u'_{S_i} f(x_{S_i})} \lr{\prod_{i=1}^\lambda u'_{T_i} f(x_{T_i})}\\
        &= \lr{\E_\PP \bra{u}\ket{v}}^2 + \E_\PP \sum_{r=1}^{\lambda p} \sum_{\substack{S,T \in \cC_\lambda\\ r(S,T)=r}} \lr{\prod_{i=1}^\lambda u'_{S_i} f(x_{S_i})} \lr{\prod_{i=1}^\lambda u'_{T_i} f(x_{T_i})}.
    \end{align}
    For any $S \in \binom{[n]}{p}$, observe that
    \begin{align}
        \alpha' := \E_P (u'_S)^2 f(x_S)^2 = \E_x \lr{1 + \frac{\beta(1-2\theta_0)}{\sqrt{\theta_0(1-\theta_0)}} f(x_S)}f(x_S)^2 = \mu + \frac{\beta(1-2\theta_0)}{\sqrt{\theta_0(1-\theta_0)}} \gamma.
    \end{align}
    If $S, T \in \binom{[n]}{p}$ satisfy $|S \cap T| = k \geq 1$, then we have
    \begin{align}
        \E_\PP u'_S f(x_S) u'_T f(x_T) = \beta^2 \E_x f(x_S)^2 f(x_T)^2 \leq \beta^2 \alpha
    \end{align}
    by \Cref{lem:momentquantities}. If $|S \cap T| = 0$, then we have
    \begin{align}
        \E_\PP u'_S f(x_S) u'_T f(x_T) = \beta^2 \E_x f(x_S)^2 f(x_T)^2 = \beta^2 \mu^2.
    \end{align}
    Hence, we obtain the second moment bound
    \begin{align}
        \E_\PP \bra{u}\ket{v}^2 &\leq \lr{\E_\PP \bra{u}\ket{v}}^2 + |\cC_\lambda|^2 \sum_{r=1}^{\lambda p} \binom{\lambda p}{r} \lr{\frac{\lambda p}{n}}^r \lr{\frac{\alpha}{\mu^2}}^r (\beta \mu)^{2\lambda}\\
        &\leq \lr{\E_\PP \bra{u}\ket{v}}^2 + |\cC_\lambda|^2 (\beta \mu)^{2\lambda} \left[\lr{1 + \frac{\alpha \lambda p}{\mu^2 n}}^{\lambda p}-1\right]
    \end{align}
    for $\lambda = o(\sqrt n)$. By \eqref{eq:uv1}, this gives variance
    \begin{align}
        \Var \bra{u}\ket{v} \leq \lr{\E \bra{u}\ket{v}}^2\left[\lr{1 + \frac{\alpha \lambda p}{\mu^2 n}}^{\lambda p}-1\right] = \lr{\E \bra{u}\ket{v}}^2 \cdot O\lr{\frac{\ell^2}{n}} = o\lr{\lr{\E \bra{u}\ket{v}}^2}.
    \end{align}
    Chebysehv's inequality thus implies concentration of $\bra{u}\ket{v}$ to its mean with probability $1-o(1)$, i.e., for any $\delta > 0$,
    \begin{align}
        \pr{\left|\frac{\bra{u}\ket{v}}{\norm{u}\norm{v}} - \left[\binom{n}{p}\beta\mu\right]^\lambda\right| \geq \delta \left[\binom{n}{p}\beta\mu\right]^\lambda} = o(1).
    \end{align}
    Finally, we union bound all bad events for the quantities $\bra{u}\ket{v}, \norm{u}^2$ and $\norm{v}^2$ from \Cref{lem:certnorm} and \Cref{lem:guidenorm} to obtain the claimed concentration. Substituting
    \begin{align}
        |\cC_\lambda| = \binom{n}{p}^\lambda \lr{1 + o(1)}
    \end{align}
    from \Cref{lem:coverlap}, we find that for any constant $\delta > 0$,
    \begin{align}
        \pr{\left|\frac{\bra{u}\ket{v}}{\norm{u}\norm{v}} - \lr{\beta\sqrt\mu}^\lambda\right| \geq \delta \lr{\beta\sqrt\mu}^\lambda} = o(1).
    \end{align}
    Taking $\beta = \tilde \Theta(\ell^{1/2-p/4}n^{-p/4})$ gives
    \begin{align}
        \pr{\frac{\bra{u}\ket{v}}{\norm{u}\norm{v}} \leq \tilde O\lr{n^{-\ell/4} \ell^{\ell(1/2p-1/4)}}} = o(1).
    \end{align}
\end{proof}

\subsection{Cutoff eigenspace proof}
Fix even $p$, write $\ell=\lambda p$ with $\ell=o(\sqrt n)$, and keep the scaling assumptions from the previous sections. Set \begin{equation}
    \beta := \frac{\epsilon}{\sqrt{\theta_0(1-\theta_0)}} \geq \frac{3\sqrt{6}}{C} \ell^{1/2 - p/4} n^{-p/4}\sqrt{\log(n)},
\end{equation} and \begin{equation}\label{eq:def-tau-spectral}
\tau := \frac{1}{2}C \beta n^{p/2} \ell^{p/2},
\end{equation} in accordance with \Cref{thm:kikuchi_bosonic}.
The previous sections defined two vectors, a certificate $\ket{v}$ and a guiding state $\ket{u}$, and established that with probability $1-o(1)$ over $\PP$, the following two statements hold:
\begin{equation}
    \frac{\bra{v} \cKt \ket{v}}{\|v\|^2} \geq \frac{4}{3} \tau,
   \qquad \text{ and } \qquad 
    \frac{\bra{\Tilde u}\ket{v}}{\norm{u}\norm{v}} > \lr{\zeta \beta \sqrt{\mu}}^{\lambda} \lr{1 - o(1)}.
\end{equation}
To show correctness of our quantum algorithm, we have to prove a subtler statement establishing that the guiding state has improved overlap not just with the certificate but with the leading eigenspace of the Kikuchi matrix. Intuitively, the leading eigenspace of $\calK$ in the planted case is the space of eigenvectors of $\cK$ with an eigenvalue that is larger than $\tau(\cK)$ in the null case. We choose an arbitrary constant, as in \Cref{thm:kikuchi_bosonic}, and define $\Pi_\geq(\cK)$ to be the projector onto the eigenspaces of $\cK$ with eigenvalues larger than $\tau$. By \Cref{thm:kikuchi_bosonic}, the largest eigenvalue of $\cK$ in the null case does not exceed $\frac{2}{3}\tau$ with high probability over $\QQ$.
Our main result is summarized in \Cref{thm:final_kikuchi_bosonic}.
\subsubsection{Step 1: Sample splitting}
\label{sec:sample_splitting}
Throughout, let $Y$ be a hypergraph instance sampled from the planted distribution as specified in \Cref{def:pbernsbm}. Fix $L:=\lceil\log n\rceil$ and define $\zeta = 1/L$. Independently of ($x, Y$), assign each $p$-set $S \in\binom{[n]}{p}$ uniformly at random to one of the $L$ batches; let $M_S \in\{0,1\}$ be the indicator that $S$ lands in Batch $B_1$, and let $M^\perp_S$ be the indicator function of the complement. (Thus $\E M_S=\zeta$, and for $S \neq T, \E\left[M_S M_T\right]=\zeta^2 ; M_S^2=M_S$.) We construct our guiding state from the first batch. 
\begin{definition}[Split guiding vector]
    Let $\ket{u''}$ be the unnormalized vector given by
    \begin{align}
        \ket{u''} = \sum_{S \in \binom{[n]}{p}} M_S\frac{Y_S - \theta_0}{\sqrt{\theta_0(1-\theta_0)}} \ket{S}.
    \end{align}
    The split guiding vector $\ket{\Tilde{u}}$ is the unnormalized vector $\Pi_\lambda \ket{u''}^{\otimes \ell/p}$.
\end{definition}
Similarly, define a split the bosonic Kikuchi matrix $\cKt$ 
on $\calT_n(\ell)$ entry-wise by \begin{equation}
        \cKt_{S,V} = \begin{cases}
            M^\perp_{\{\mu_1, \dots, \mu_p\}} \frac{Y_{\{\mu_1, \dots, \mu_p\}} - \theta_0}{\sqrt{\theta_0(1-\theta_0)}} \qquad \text{ if } (\mu_1, \dots, \mu_p) = S \ominus V, \\
            0 \qquad \qquad \qquad \qquad \qquad \ \ \ \text{ otherwise.}
        \end{cases}
    \end{equation}
Here and in what follows, we absorb the distribution of the random uniform splitting into the notation $\PP(x)$. Under the probability distribution $\PP(x)$ conditioned on $x$, the matrix and guiding state are independent. Due to independence of $M_S$,
the split Kikuchi matrix satisfies \begin{equation}
\label{eq:batched_ineqs}
    \frac{\bra{v} \cK \ket{v}}{\|v\|^2} \geq \frac{4}{3} (1-\zeta)\tau,
   \qquad \text{ with probability } 1-o(1) \text{ over } \PP.  
\end{equation}

\subsubsection{Step 2: Leading eigenspace mass}
We record the simple observation that $\ket{v}$ has large support on the leading eigenspace of $\cKt$.
\begin{theorem}
\label{thm:markov_arg}
    Let $\Pi_\geq(\cKt)$ denote the projector onto the eigenvectors of $\cKt$ with eigenvalues at least $\tau$.
    Then with probability $1-o(1)$ over $\PP$, \begin{equation}
        \frac{\bra{v} \Pi_\geq(\cKt) \ket{v}}{\braket{v}{v}} \geq \frac{1-4\zeta}{6/C' - 1} = \Omega(1).
    \end{equation}
\end{theorem}
\begin{proof}
    Conditioned on labels x, $\E_{\PP(x)}[A_S]=\beta\,f(x_S)$ with $|f|\le 1$, hence we have the loose upper bound $||\E_{\PP(x)}[\cKt]|| \leq \beta n^{p/2}\ell^{p/2}$. The centered random variable $\cKt - \E_{\PP(x)}[\cKt]$ is bounded by $\frac{2}{3}\tau$ whp, which is $\leq \frac{1}{3} C\beta n^{p/2}\ell^{p/2}$ by choice of $\beta$ in \Cref{thm:kikuchi_bosonic}) except with probability $o(1)$ following the analysis of \Cref{lem:bosonic_nullknorm}. Hence $\cKt \leq (1 + C/3)\beta n^{p/2}\ell^{p/2} =(2/C + 2/3) \cdot \tau$ with probability $1-o(1)$. The claim follows from \cref{eq:batched_ineqs} and a Markov-style argument applied to the eigenvalues of $\cKt$, absorbing the negligible loss of $1-\zeta$ due to sample splitting into the constant.
\end{proof}
\subsubsection{Step 3: Directional unbias}
Our third main ingredient states that the guiding vector $\ket{\Tilde u}$ is directionally unbiased. 
\begin{theorem}\label{thm:2nd_markov_arg}
Let $p\ge 1$, $\lambda\ge 1$, and $\ell=\lambda p$.  Let $\ket{s}$ be a unit vector (indexed by
$\mathcal T_n(\ell)$) that is independent of $\ket{\Tilde u}$ under $\PP(x)$.  Then
\[
\braket{s}{\Tilde u}\ \ge\ \frac12\,\zeta^\lambda \beta^\lambda \braket{s}{v}
\]
except with probability (over $\PP$) of at most
\begin{equation}\label{eq:XXX-basic}
\mathrm{FAIL}\ :=\ \frac{4\,\lambda!}{\zeta^{2\lambda}\beta^{2\lambda}\braket{s}{v}^2}\,.
\end{equation}
\end{theorem}
To prove this, we start with an auxiliary lemma.
\begin{lemma}\label{lem:main_technical}
Let $p \ge 1$ and $\lambda \ge 1$, and fix integers $a,b \ge 0$ with $a+b=\lambda$. Let $(A_S)_{S}$ be i.i.d. for $S \in \binom{n}{p}$. For distinct $p$-sets $S_1,\dots,S_a, S'_1,\dots,S'_b, S''_1,\dots,S''_b$, define \begin{equation}
W(a,b) = \Cov_{\PP(x)} \left[ \prod_{i=1}^{a} A_{S_i} \prod_{j=1}^{b} A_{S'_j}, \ \prod_{i=1}^{a} A_{S_i} \prod_{j=1}^{b} A_{S''_j} \right].
\end{equation}  
Then:
\begin{align}
\E_x W(a,b) &=
\begin{cases}
0, & a=0,1,\dots,\lambda-1,\\[2pt]
1-(\beta^2\mu)^\lambda, & a=\lambda \ (b=0),
\end{cases}
\label{eq:main-Ex}
\\[6pt]
\Var_x W(a,b) &=
\begin{cases}
(\beta^2 \mu)^{2b}\left( (1+\eps^2 \mu c^2)^{a} + (\beta^4 \alpha)^{a} - 2 \left[\beta^2(\mu+\eps \gamma c)\right]^{a} \right), & a=0,1,\dots,\lambda-1,\\[4pt]
(1+\eps^2 \mu c^2)^{\lambda} + (\beta^4 \alpha)^{\lambda} - 2 \left[\beta^2(\mu+\eps \gamma c)\right]^{\lambda} - \left(1-(\beta^2 \mu)^{\lambda}\right)^2, & a=\lambda \ (b=0).
\end{cases}
\label{eq:main-Varx}
\end{align}
\end{lemma}
\begin{proof}
Recall the dictionary of $x$-averaged moments:
\begin{align}
\E_x \E_{\PP(x)}[A_S] &= 0, \\
\E_x \E_{\PP(x)}[A_S]^2 &= \beta^2 \mu, \\
\E_x \E_{\PP(x)}[A_S]^4 &= \beta^4 \alpha, \\
\E_x \E_{\PP(x)}[A_S^2] &= 1, \\
\E_x \E_{\PP(x)}[A_S^2]^2 &= 1 + \eps^2 \mu c^2, \\
\E_x \E_{\PP(x)}[A_S^2 \cdot (\E_{\PP(x)}[A_S])^2] &= \beta^2 (\mu + \eps \gamma c),
\end{align}
where $c = \frac{1 - 2\theta_0}{\theta_0(1-\theta_0)}$ and $\mu,\alpha,\gamma$ depend only on $(k,p)$.

Let $U=\prod_{i=1}^{a} A_{S_i}$, $V'=\prod_{j=1}^{b} A_{S'_j}$, and $V''=\prod_{j=1}^{b} A_{S''_j}$. Because the sets are disjoint, independence under $\PP(x)$ gives
\begin{equation}
\E_{\PP(x)}[U^2 V' V''] = \prod_{i=1}^a \E_{\PP(x)}[A_{S_i}^2] \prod_{j=1}^b \E_{\PP(x)}[A_{S'_j}] \E_{\PP(x)}[A_{S''_j}],
\end{equation}
\begin{equation}
\E_{\PP(x)}[U V'] = \prod_{i=1}^a \E_{\PP(x)}[A_{S_i}] \prod_{j=1}^b \E_{\PP(x)}[A_{S'_j}], \qquad
\E_{\PP(x)}[U V''] = \prod_{i=1}^a \E_{\PP(x)}[A_{S_i}] \prod_{j=1}^b \E_{\PP(x)}[A_{S''_j}].
\end{equation}
It follows that
\begin{equation}\label{eq:W-factorization}
W(a,b)(x) = \left(\prod_{j=1}^b \E_{\PP(x)}[A_{S'_j}] \E_{\PP(x)}[A_{S''_j}]\right)\left( \prod_{i=1}^a \E_{\PP(x)}[A_{S_i}^2] - \prod_{i=1}^a \E_{\PP(x)}[A_{S_i}]^2 \right).
\end{equation}

Averaging over $x$ factorizes across distinct $p$-sets, so
\begin{equation}\label{eq:Ex-general}
\E_x W(a,b) = \left(\E_x \E_{\PP(x)}[A_S]\right)^{2b}\left( \left(\E_x \E_{\PP(x)}[A_S^2]\right)^a - \left(\E_x \E_{\PP(x)}[A_S]^2\right)^a \right).
\end{equation}
Using the dictionary, $\E_x \E_{\PP(x)}[A_S]=0$, $\E_x \E_{\PP(x)}[A_S^2]=1$, and $\E_x \E_{\PP(x)}[A_S]^2=\beta^2\mu$. Thus if $b \ge 1$ the expectation is zero, while if $b=0$ we obtain $\E_x W(\lambda,0)=1-(\beta^2\mu)^\lambda$, proving \eqref{eq:main-Ex}.

For the variance, write $W=G \cdot H$ where
\begin{equation}
G(x) = \prod_{j=1}^{b} \E_{\PP(x)}[A_{S'_j}] \E_{\PP(x)}[A_{S''_j}], \qquad
H(x) = \prod_{i=1}^{a} \E_{\PP(x)}[A_{S_i}^2] - \prod_{i=1}^{a} \E_{\PP(x)}[A_{S_i}]^2.
\end{equation}
The variables $G$ and $H$ depend on distinct sets and are independent under $\E_x$, so
\begin{equation}
\E_x[W^2] = \E_x[G^2] \cdot \E_x[H^2].
\end{equation}
Independence again gives
\begin{equation}
\E_x[G^2] = \left(\E_x \E_{\PP(x)}[A_S]^2\right)^{2b} = (\beta^2 \mu)^{2b},
\end{equation}
while
\begin{equation}
\E_x[H^2] = \left(\E_x \E_{\PP(x)}[A_S^2]^2\right)^a - 2\left(\E_x \E_{\PP(x)}[A_S^2 \cdot (\E_{\PP(x)}[A_S])^2]\right)^a + \left(\E_x \E_{\PP(x)}[A_S]^4\right)^a.
\end{equation}
Substituting the dictionary values yields
\begin{equation}
\E_x[W(a,b)^2] = (\beta^2 \mu)^{2b}\left((1+\eps^2 \mu c^2)^a - 2[\beta^2(\mu+\eps \gamma c)]^a + (\beta^4 \alpha)^a\right).
\end{equation}
Finally, subtracting $(\E_x W(a,b))^2$ (from \eqref{eq:Ex-general}) gives
\begin{equation}
\Var_x W(a,b) = (\beta^2 \mu)^{2b}\left((1+\eps^2 \mu c^2)^a + (\beta^4 \alpha)^a - 2[\beta^2(\mu+\eps \gamma c)]^a\right)
\end{equation}
when $a \le \lambda-1$ (so $b \ge 1$), while for $a=\lambda$ and $b=0$ we obtain
\begin{equation}
\Var_x W(\lambda,0) = (1+\eps^2 \mu c^2)^{\lambda} + (\beta^4 \alpha)^{\lambda} - 2[\beta^2(\mu+\eps \gamma c)]^{\lambda} - \left(1-(\beta^2 \mu)^{\lambda}\right)^2,
\end{equation}
which is exactly \eqref{eq:main-Varx}.
\end{proof}
\Cref{lem:main_technical} is the key ingredient in the below proof.
\begin{proof}[Proof of \Cref{thm:2nd_markov_arg}.]
Fix $x\in [k]^n$ and work under the conditional probability $\PP(x)$.  We first work with the full guiding state $\ket{u}$; the statement transfers to $\ket{\Tilde u}$ by taking the expectation over the splitting process (using $A_S \mapsto M_SA_S$ and $\E M_S=\zeta$, leading to $\beta \mapsto \zeta \beta$ in all dictionary entries).
By construction,
$\E_{\PP(x)}\ket{u}=\beta^\lambda\ket{v}$, so
\begin{equation}\label{eq:mean}
\E_{\PP(x)}\braket{s}{u}=\beta^\lambda\braket{s}{v}.
\end{equation}
Thus, for the one-sided event
\[
\mathcal E_x\ :=\ \Big\{\braket{s}{u}<\tfrac12\,\beta^\lambda\braket{s}{v}\Big\},
\]
Chebyshev’s inequality (applied to the centered variable
$\braket{s}{u}-\E_{\PP(x)}\braket{s}{u}$) gives
\begin{equation}\label{eq:cheb}
\Pr_{\PP(x)}\big(\mathcal E_x\big)\ \le\ 
\frac{4\,\Var_{\PP(x)}\braket{s}{u}}{\beta^{2\lambda}\,\braket{s}{v}^{2}}\,.
\end{equation}
We next bound the conditional variance.  Writing $u_T=\prod_{i=1}^\lambda A_{S_i}$ for
$T=(S_1,\dots,S_\lambda)\in\mathcal T_n(\ell)$, we have
\begin{equation}\label{eq:var-expansion}
\Var_{\PP(x)}\braket{s}{u}
=\sum_{T,V\in\mathcal T_n(\ell)} s_T s_V\ \Cov_{\PP(x)}(u_T,u_V)
\ \le\ \sum_{T\in\mathcal T_n(\ell)} s_T^{2}\ \sum_{V\in\mathcal T_n(\ell)} \Cov_{\PP(x)}(u_T,u_V),
\end{equation}
by Cauchy-Schwarz.  Taking expectation over $x$ and using Lemma~\ref{lem:main_technical}
(the terms with $a<\lambda$ vanish in expectation, and the $a=\lambda$ term equals $W(\lambda,0)$),
\begin{equation}\label{eq:Ex-var}
\E_x\,\Var_{\PP(x)}\braket{s}{u}
=\lambda!\,\Big(\sum_{T} s_T^{2}\Big)\,\E_x W(\lambda,0)
\ \le\ \lambda!\,.
\end{equation}
Finally, average \eqref{eq:cheb} over $x$ and use \eqref{eq:Ex-var}:
\[
\Pr_{x,\PP(x)}\!\Big(\braket{s}{u}<\tfrac12\,\beta^\lambda\braket{s}{v}\Big)
=\E_x\big[\Pr_{\PP(x)}(\mathcal E_x)\big]
\ \le\ \frac{4}{\beta^{2\lambda}\,\braket{s}{v}^{2}}\ \E_x\,\Var_{\PP(x)}\braket{s}{u}
\ \le\ \frac{4\,\lambda!}{\beta^{2\lambda}\,\braket{s}{v}^{2}}\,,
\]
as claimed. Setting $\beta \mapsto \zeta \beta$ completes the proof. 
\end{proof}

\begin{remark}
The bound \eqref{eq:XXX-basic} can be almost certainly improved, for example by a refined variance control using geometric decay in the overlap levels and the $a=1$ term domination, as in the proof of theorem 38 in \cite{schmidhuber2025quartic}.
\end{remark}
\subsection{Putting everything together}
\label{sec:putting_together}
We now prove \Cref{thm:final_kikuchi_bosonic}, and use that to establish \Cref{thm:kikuchi_q} stated in \Cref{sec:Quantum}.
\begin{proof}[Proof of \Cref{thm:final_kikuchi_bosonic}]
With $\Pi_\geq$ as defined in the theorem statement, define the normalized quantum state $\ket{s} = \frac{1}{||\Pi_\geq \ket{v}||}\Pi_\geq \ket{v}$. By \Cref{thm:markov_arg} and \Cref{lem:certnorm}, $\braket{s}{v} = ||\Pi_\geq \ket{v}||\geq \sqrt{C'/\sqrt{12} \binom{n}{p}^\lambda \mu^\lambda}$ except with probability $o(1)$. Following the proof of \Cref{thm:2nd_markov_arg}, we thus have \begin{equation}\braket{s}{\Tilde u} \geq \frac{1}{2}\left(\beta \zeta \right)^\lambda \braket{s}{v} \geq  \frac{\sqrt{C'}}{12}\left(\beta \zeta \right)^\lambda \binom{n}{p}^{\lambda/2} \mu^{\lambda/2}\end{equation}
except with probability (over $\PP$) of at most
\begin{equation}
\mathrm{FAIL}\ =\ O \left( \frac{\lambda!}{\beta^{2\lambda}\binom{n}{p}^{\lambda}}\, \right) = O\left(\frac{\lambda!}{n^\lambda}\right) = O(n^{-\lambda/2}) = o(1).
\end{equation}
Because $\Pi_\geq$ is a positive semi-definite operator, and using \Cref{lem:guidenorm} appropriately re-normalized with $\E[M_S] = \zeta$, this implies \begin{equation}
   \frac{ \bra{\Tilde u}\Pi_\geq(\cKt)\ket{\Tilde u}}{   \braket{\Tilde u}{\Tilde u}} \geq \frac{ \bra{\Tilde u}\ketbra{s}{s}\ket{\Tilde u}}{   \braket{\Tilde u}{\Tilde u}} \geq \frac{C' (\beta \zeta)^{2\lambda} \binom{n}{p}^\lambda \mu^\lambda}{12 \zeta^\lambda \binom{n}{p}^\lambda } \geq \Omega \left(\zeta^\lambda \mu^\lambda \beta^{2\lambda} \right) 
\end{equation} except with probability $o(1)$, where we have union bounded away all bad events for the quantities $\bra{s}\ket{v}, \bra{s}\ket{\Tilde u}, \norm{\Tilde u}^2$ and $\norm{v}^2$. The statement follows by plugging in $\zeta = \frac{1}{\lceil\log n\rceil}$ and $\beta = \Tilde\Omega\left(\ell^{1/2 - p/4} n^{-p/4} \right)$.
\end{proof}

\begin{proof}[Proof of \Cref{thm:kikuchi_q}]
    As shown already previously, sample slitting reduces the planted energy in \Cref{cor:planted_energy} by at most a factor of $(1-\zeta)$, which is well within the threshold gap in \Cref{thm:kikuchi_bosonic}. By \Cref{lem:efficient_guide_prep}, preparing the guiding state requires time $\tilde O\lr{\ell n^p}$. By the proof of \Cref{thm:final_kikuchi_bosonic}, Quantum Phase Estimation combined with Amplitude Amplification requires a number of repetitions scaling as $ O\lr{\log n^{\ell/2} \cdot \exp(O(\ell)) \cdot n^{\ell/4} \cdot \ell^{\frac{\ell}{4}-\frac{\ell}{2p}}}$ to find an eigenvector certifying the community structure. Implementing the sparse oracle for QPE costs $O(n^p \ell \log n)$ as stated in \Cref{sec:space_adv}, which combines additively with the state preparation step. Putting everything together establishes the claimed gate cost. The quantum and classical space requirements follow from the discussion in \Cref{sec:space_adv}.
\end{proof}

\subsection{Guiding state: efficient preparation}
\label{sec:guide_3}\label{sec:symmetrization}
We now show that the guiding state can be efficiently prepared. 
\begin{lemma}[Guiding state preparation]
\label{lem:efficient_guide_prep}
    Assume $\ell = o(\sqrt n)$ and let $\ket{u}$ be the unnormalized guiding vector of \Cref{def:guide_vector}. With probability $1-o(1)$, the normalized guiding state $\ket{\tilde u} = \ket{u}/\norm{u}$ can be prepared to trace distance $\epsilon$ with cost
    \begin{align}
        \tilde O\lr{\ell n^p \log \frac{1}{\epsilon}}.
    \end{align}
\end{lemma}
\begin{proof}
    Let $\ket{w'} = \sum_{S\in \binom{[n]}{p}} w'_S \ket{S}$ be a normalized vector. The success probability of preparing the state $\ket{w}$ proportional to $\Pi_\lambda \norm{w'}^{\otimes \lambda}$ is then
    \begin{align}
        \eta = \norm{\Pi_\lambda \ket{w'}^{\otimes \lambda}}^2 = \sum_{(S_1,\dots,S_\lambda)\in\cC_\lambda} \prod_{t=1}^\lambda |w'_{S_t}|^2.
    \end{align}
    Hence, if $C_{w'} = O(n^p)$ is the gate cost to prepare $\ket{w'}$ and $C_\Pi = \tilde O(\lambda p \log n)$ is the cost to test collision-freeness, the total cost of preparing the state $\ket{w}$ to $\epsilon$ trace distance is
    \begin{align}\label{eq:stateprep}
        \tilde O\lr{\frac{\lambda C_{w'} + C_\Pi}{\sqrt{\eta}} \log \frac{1}{\epsilon}}
    \end{align}
    via amplitude amplification. We proceed to show concentration bounds on $\eta$ for vector $\ket{u'}$.
    
    Recall that $\ket{u} = \Pi_\lambda \ket{u'}^{\otimes \lambda}$. As computed in \Cref{lem:guidenorm}, for any $\delta > 0$,
    \begin{align}
        \pr{\left|\norm{u}^2 - \binom{n}{p}^\lambda\right| \geq \delta \binom{n}{p}^\lambda} = o(1)
    \end{align}
    for $\ell = o(\sqrt n)$. Hence, it suffices to show that $\norm{u'}^2 = \binom{n}{p}\lr{1 \pm o(1/\lambda)}$ with high probability.
    The first moment is
    \begin{align}
        \E \norm{u'}^2 &= \E_x \sum_S \lr{1 + \epsilon f(x_S)\frac{1-2\theta_0}{\theta_0(1-\theta_0)}} = \binom{n}{p}
    \end{align}
    and the second moment is
    \begin{align}
        \E \norm{u'}^4 &= \frac{1}{[\theta_0(1-\theta_0)]^2}\lr{\E \sum_{S} (Y_S-\theta_0)^4 + \sum_{S\neq T}(Y_S-\theta_0)^2(Y_T-\theta_0)^2}\\
        &= \binom{n}{p}\frac{\theta_0^4(1-\theta_0) + \theta_0(1-\theta_0)^4}{{[\theta_0(1-\theta_0)]^2}}\nonumber\\
        &\qquad + \frac{1}{[\theta_0(1-\theta_0)]^2}\E_x \sum_{S\neq T} [\theta_0(1-\theta_0) + \epsilon f(x_S)(1-2\theta_0)][\theta_0(1-\theta_0) + \epsilon f(x_T)(1-2\theta_0)]\\
        &= \frac{1}{[\theta_0(1-\theta_0)]^2}\lr{\binom{n}{p}\lr{\theta_0^4(1-\theta_0) + \theta_0(1-\theta_0)^4 + \left[\binom{n}{p}-1\right][\theta_0(1-\theta_0)]^2}}\\
        &= \binom{n}{p}^2\lr{1 + O(n^{-p})}.
    \end{align}
    Hence, we obtain variance 
    \begin{align}
        \Var \norm{u'}^2 = O(n^p).
    \end{align}
    Taking $\ket{w'} = \ket{u'}/\sqrt{\binom{n}{p}}$ so $\Var \norm{w'}^2 = O(n^{-p})$, Chebyshev's inequality gives for any $\delta > 0$ that
    \begin{align}
        \pr{\left|\norm{w'}^2 - 1\right| \geq \frac{\delta}{\lambda}} = O\lr{\frac{\lambda^2}{n^p}}.
    \end{align}
    We complete the proof by taking $\lambda = o(\sqrt n)$.
\end{proof}

\end{document}